\newif\ifarxiv
\arxivtrue

\ifarxiv
\documentclass[12pt]{article}
\usepackage{a4wide}
\usepackage[numbers,sort]{natbib}
\usepackage{authblk}
\author[1]{Matthias Bentert}
\author[2]{Chen Avin}
\author[1]{Stefan Schmid}
\affil[1]{Technische Universität Berlin, Germany}
\affil[2]{Ben-Gurion University of the Negev, School of Electrical and Computer Engineering, Israel}
\else
\documentclass[sigconf]{acmart}
\author{Matthias Bentert}
\orcid{0009-0009-0705-972X}
\affiliation{%
\institution{Technische Universität Berlin}
\city{Berlin}
\country{Germany}}
\email{bentert@tu-berlin.de}
\author{Chen Avin}
\orcid{0000-0002-6647-8002}
\affiliation{%
\institution{Ben-Gurion University of the Negev, School of Electrical and Computer Engineering}
\city{Beersheba}
\country{Israel}}
\email{avin@bgu.ac.il}
\author{Stefan Schmid}
\orcid{0000-0002-7798-1711}
\affiliation{%
\institution{Technische Universität Berlin}
\city{Berlin}
\country{Germany}}
\email{stefan.schmid@tu-berlin.de}


\fi
\usepackage{amsmath}
\usepackage{amsthm}
\usepackage{todonotes}
\usepackage{framed}
\usepackage{tabularx}
\usepackage{tcolorbox}
\usepackage{tikz}
\usetikzlibrary{calc,shapes}
\usepackage{subcaption}
\usepackage{dsfont}
\usepackage[capitalize,noabbrev]{cleveref}
\usepackage{balance}

\newtheorem{theorem}{Theorem}
\newtheorem{observation}{Observation}
\newtheorem{lemma}{Lemma}
\newtheorem{proposition}{Proposition}

\newcommand{\thr}{\textsc{Demand-Aware Weak Throughput}}
\newcommand{\strong}{\textsc{Demand-Aware Throughput}}
\newcommand{\dt}{\textsc{Demand-Aware Direct Throughput}}
\newcommand{\di}{\textsc{Demand-Aware Weak Direct Throughput}}

\newcommand{\N}{\ensuremath{2n-1}}
\newcommand{\Np}{\ensuremath{(2n-1)}}
\newcommand{\ns}{\ensuremath{n^*}}
\newcommand{\NN}{\ensuremath{\mathds{N}}}
\newcommand{\new}{weak throughput}
\newcommand{\old}{throughput}
\newcommand{\direct}{direct throughput}

\DeclareMathOperator{\len}{len}
\DeclareMathOperator{\con}{cont}

\newlength{\RoundedBoxWidth}
\newsavebox{\GrayRoundedBox}
\ifarxiv
\newenvironment{GrayBox}[1]%
   {\setlength{\RoundedBoxWidth}{.93\textwidth}
    \def\boxheading{#1}
    \begin{lrbox}{\GrayRoundedBox}
       \begin{minipage}{\RoundedBoxWidth}}%
   {   \end{minipage}
    \end{lrbox}
    \begin{center}
    \begin{tikzpicture}%
       \node(Text)[draw=black!20,fill=white,rounded corners,%
             inner sep=2ex,text width=\RoundedBoxWidth]%
             {\usebox{\GrayRoundedBox}};
        \coordinate(x) at (current bounding box.north west);
        \node [draw=white,rectangle,inner sep=3pt,anchor=north west,fill=white] 
        at ($(x)+(6pt,.75em)$) {\boxheading};
    \end{tikzpicture}
    \end{center}}

\newenvironment{defproblemx}[2][]{\noindent\ignorespaces%
    \FrameSep=6pt%
    \parindent=0pt%
    \vspace*{-1.5em}
    \begin{GrayBox}{#2}
    \begin{tabular*}{\textwidth}{@{\hspace{.1em}} >{\itshape} p{1.7cm} p{0.86\textwidth} @{}}%
}{
    \end{tabular*}%
    \end{GrayBox}%
    \ignorespacesafterend
}
\else
\newenvironment{GrayBox}[1]%
   {\setlength{\RoundedBoxWidth}{.45\textwidth}
    \def\boxheading{#1}
    \begin{lrbox}{\GrayRoundedBox}
       \begin{minipage}{\RoundedBoxWidth}}%
   {   \end{minipage}
    \end{lrbox}
    \begin{center}
    \begin{tikzpicture}%
       \node(Text)[draw=black!20,fill=white,rounded corners,%
             inner sep=2ex,text width=\RoundedBoxWidth]%
             {\usebox{\GrayRoundedBox}};
        \coordinate(x) at (current bounding box.north west);
        \node [draw=white,rectangle,inner sep=3pt,anchor=north west,fill=white] 
        at ($(x)+(6pt,.75em)$) {\boxheading};
    \end{tikzpicture}
    \end{center}}

\newenvironment{defproblemx}[2][]{\noindent\ignorespaces%
    \FrameSep=6pt%
    \parindent=0pt%
    \vspace*{-1.5em}
    \begin{GrayBox}{#2}
    \begin{tabular*}{\textwidth}{@{\hspace{.1em}} >{\itshape} p{1cm} p{0.81\textwidth} @{}}%
}{
    \end{tabular*}%
    \end{GrayBox}%
    \ignorespacesafterend
}
\fi

\newcommand{\defproblem}[3]{
  \begin{defproblemx}{#1}
    Input:  & #2 \\
    Question: & #3
  \end{defproblemx}
}%

\title{A Separation Between Optimal Demand-Oblivious and Demand-Aware Network Throughput}
\date{}

\ifarxiv
\begin{document}
\maketitle
\fi

\begin{abstract}
The performance of distributed applications often critically depends on the interconnecting network or more specifically on its throughput: how fast data can be carried across a network.
Over the last years, great progress has been made in understanding 
\emph{demand-oblivious} throughput: how fast a given demand matrix describing pairwise communication
requirements, can be served on a \emph{given} network.
However, surprisingly little is known today about
the achievable \emph{demand-aware} throughput: the throughput on a network topology which
can be \emph{optimized} toward the demand. Such demand-aware networks have recently gained
popularity in datacenters and are enabled by emerging reconfigurable optical technologies.

In this paper, we are interested in both the achievable demand-aware throughput bounds 
as well as in the computational complexity of finding a throughput-optimizing network topology.
We take a systematic approach and investigate four variants of demand-aware throughput: we analyze, and derive bounds for, 
two definitions of throughput, the classic throughput usually considered in the literature, and 
a new generalized definition which we call \new; for each of them, we consider two routing models, a \emph{direct}
one where demand can only be served on a single hop and a general one where multi-hop routing is allowed.

Our main result is a separation result which solves an open problem in the literature about the classical throughput definition,
showing that demand-aware topologies can outperform demand-oblivious topologies even in the worst case:
the demand-aware throughput asymptotically approaches at least~$\frac{5}{8}$, while it is known 
that the demand-oblivious throughput is at most~$\frac{n}{2n-1} \approx \frac{1}{2}$. 
In terms of computational complexity, we show that computing the demand-aware weak throughput is NP-hard, 
but computing the demand-aware direct throughput and the demand-aware weak direct throughput are both polynomial-time solvable.
We leave it as an open problem whether computing the demand-aware throughput 
is polynomial-time solvable or NP-hard, but we conjecture that it is NP-hard.
\end{abstract}

\ifarxiv
\newpage
\else
\begin{document}
\maketitle
\fi

\section{Introduction}

Many data-intensive distributed applications, especially those related
to high-performance computing and machine learning, critically depend on a high-throughput interconnecting communication
network. Accordingly, over the last decades, significant efforts have been made
to design novel datacenter networks which provide high capacity at low cost.
To this end, researchers intensively studied how to define and measure the \emph{throughput of a network}~\cite{jyothi2016measuring,dc-throughput},
how to model its impact on application
performance~\cite{faizian2017throughput,mogul2012we}, 
and how to improve it~\cite{addanki2025vermilion,al2008scalable,amir2022optimal,baral2026universal,li2019hpcc,mellette2017rotornet,singla2014high}. A particularly intriguing solution are demand-aware networks (also known as topology engineering in the literature)~\cite{avin2025revolutionizing,avin2020demand,avin2022demand,firefly,projector,helios,griner2021cerberus,poutievski2022jupiter,zerwas2023duo}: networks whose topology can be optimized towards 
the demand using emerging optical switching technologies.

However, defining and analyzing the throughput of a demand-aware network topology is surprisingly subtle.
Existing literature so far mostly revolves around the  
\emph{demand-oblivious} throughput: 
design a network topology (a graph)~$G$ (with given link capacities) that can serve well \emph{any} 
(doubly stochastic) demand matrix~$\mathcal{M}$, describing pairwise communication requirements.
More specifically, the throughput of a (directed) graph $G$ with respect to a demand matrix $\mathcal{M}$ is defined by the maximum  factor by which the rates specified in the demand matrix must be scaled (or equivalently, the time horizon of the transmission increased)  such that they can be served on the topology without violating its link capacities~\cite{jain2014maximizing,jyothi2016measuring,yuan2014lfti,zerwas2023duo,griner2021cerberus,addanki2025vermilion,addanki2023mars,KCKL05}. The throughput of a graph is defined as its worst-case throughput
over all possible demand matrices.
It is already known that the throughput of demand-oblivious networks connecting $n$ nodes is 
at most~$\frac{n}{2n-1} \approx \frac{1}{2}$~\cite{KCKL05}.
Empirical studies suggest that demand-aware networks can achieve a higher throughput.
However, whether this also holds under worst-case demand matrices has been an open question for several years~\cite{griner2021cerberus,addanki2025vermilion,addanki2023mars,farrington2013multiport}.
We study this fundamental question.
Particularly: is there a separation between the worst case demand-aware throughput and the demand-oblivious throughput? 

To shed light on the demand-aware throughput, we take a systematic approach and study both the achievable bounds as well as the computational complexity of optimizing the topology. Moreover, we consider two different throughput definitions and two different routing models, that is, four scenarios in total. 
In particular, in addition to the classic throughput definition in the literature mentioned above, we also introduce an intuitive and more general
definition which we call the \new{} of a network: in a nutshell (details will follow), the weak throughput measures the highest possible percentage of demands that can be served.
In terms of routing models, we consider a multi-hop scenario, where the demand can be routed
arbitrarily over the network, and a single-hop scenario, where the demand can be served only across a single link---as is sometimes the case in optical networks~\cite{hall2021survey}.
We refer to the latter as \emph{direct} routing.
Before presenting our contributions in detail, we introduce our model more formally.

\subsection{Formal Model}
\label{sec:prelim}

For a positive integer~$n$, we use~$[n]$ to denote the set~$\{1,2,\ldots,n\}$.
We use~$\boldsymbol{1}_x$ for a Boolean expression~$x$ as the indicator function: ``1~if~$x$ holds and~$0$ otherwise''.
For a sequence~$\sigma = (\phi_1,\phi_2,\ldots,\phi_{\ell})$, we use~$x \in \sigma$ as a shorthand for~$x \in \{\phi_1,\phi_2,\ldots,\phi_{\ell}\}$.

We model networks as multigraphs, that is, graphs can contain self-loops and parallel arcs.
If each vertex in a directed graph~$G$ has exactly~$r$ incoming and outgoing arcs, then we call~$G$ a \emph{directed~$r$-regular graph}.
A path~$P$ in a directed graph is a sequence of arcs~$((u_1,v_1),(u_2,v_2),\ldots,(u_\ell,v_\ell))$ such that for all~${i \in [\ell-1]}$ it holds that~$v_i = u_{i+1}$.
We say that~$P$ is a path from~$u_1$ to~$v_\ell$.
The \emph{length}~$\len(P)=\ell$ of a path~$P$ is the number of arcs in it and we only consider paths containing at least one arc in this work.
This is consistent with the existing literature, where demand from a node~$v$ to itself has to use some arcs (possibly a self-loop).
We mention in passing that while our specific constructions do use self-loops, all of our results can also be transferred to the setting where demand from a node to itself can be served by empty paths.
The throughput bounds (and a construction in a hardness proof later on) will then of course be slightly different.

A matrix is doubly stochastic if each entry is non-negative and each row and each column sums to~$1$.
As in the existing literature, we will assume that all input matrices are doubly stochastic.
For a matrix~$\mathcal{M} = (a_{i,j})_{i,j \in [n]}$ and a number~$x$, we use~$x \mathcal{M}$ to denote the matrix, where each entry is scaled by a factor of~$x$, that is, the matrix~$(x a_{i,j})_{i,j \in [n]}$.
We will often match the~$i$-th row and~$i$-th column of a matrix to the~$i$-th vertex in a graph.
For notational convenience, we will often use~$[n]$ as the set of vertices to easily index the rows and columns of the matrix with the names of the vertices.
Given a matrix~${\mathcal{M} = (a_{i,j})_{i,j \in [n]}}$ and a directed graph~$G$ with arc capacities where the sum of capacities entering and leaving each vertex is~$1$, we say that~$G$ can \emph{host}~$\mathcal{M}$ if there is a flow (a collection~${\mathcal{P} = \{(P_1,d_1),(P_2,d_2),\ldots,(P_{\ell},d_{\ell})\}}$ of $\ell$ pairs---where each~$P_i$ is a directed (non-empty) path in~$G$ from some vertex~$s_i$ to some (not necessarily different) vertex~$t_i$ and~$d_i$ is a number) such that~$\mathcal{P}$ is a multi-commodity flow in~$G$ satisfying the demand requirements of~$\mathcal{M}$ and not violating the capacities of~$G$.
If the graph is~$r$-regular and each arc has capacity~$\frac{1}{r}$, then this is equivalent to setting each capacity to~$1$ and requiring that~$G$ can satisfy the multicommodity flow described by~$r \mathcal{M}$.
More formally, the following holds.
\begin{itemize}
    \item For all~$(u,v) \in V \times V$, $\sum_{i \in [\ell]} (d_i \cdot \boldsymbol{1}_{s_i=u} \cdot \boldsymbol{1}_{t_i=v}) = r \cdot a_{u,v}$ and
    \item for all arcs~$e = (u,v)$ in~$G$, it holds that $\sum_{i \in [\ell]} (\boldsymbol{1}_{e \in P_i} \cdot d_i) \leq 1$.
\end{itemize}

We say that~$G$ can \emph{directly host}~$\mathcal{M}$ if~$G$ can host~$\mathcal{M}$ in such a way that all paths in~$\mathcal{P}$ only consist of a single arc each.
The \emph{throughput} of~$G$ with respect to~$\mathcal{M}$ is the largest number~$\theta$ such that~$G$ can host the matrix~$\theta \mathcal{M}$ and the direct throughput of~$G$ with respect to~$\mathcal{M}$ is the largest number~$\theta'$ such that~$G$ can directly host the matrix~$\theta' \mathcal{M}$.\footnote{
There also exist cut-based measures to capture the throughput or ``capacity
of a network'' in the literature~\cite{al2008scalable,greenberg2011vl2,ahn2009hyperx,curtis2012rewire}. However, these have shown to be suboptimal~\cite{jyothi2016measuring,dc-throughput,yuan2014lfti} and have become less popular recently.} 
The \emph{\new} of~$G$ with respect to~$\mathcal{M}$ is the largest number~$\eta$ such that~$G$ can host a matrix~$\mathcal{M'}= (a'_{i,j})_{i,j \in [n]}$ where~$a'_{i,j} \leq a_{i,j}$ for all~$i,j \in [n]$ and~$\frac{\sum_{(i,j) \in V \times V} a'_{i,j}}{\sum_{(i,j) \in V \times V} a_{i,j}} = \eta$.
We also say that~$G$ can host an~$\eta$-fraction of~$\mathcal{M}$.
Note that since we assume that the input matrix is doubly stochastic, the weak throughput can equivalently be defined by~$\frac{\sum_{(i,j) \in V \times V} a'_{i,j}}{n} = \eta$.
Finally, the direct weak throughput of~$G$ with respect to~$\mathcal{M}$ is the largest number~$\eta$ such that~$G$ can directly host an~$\eta$-fraction of~$\mathcal{M}$ and the (weak) throughput of a graph is the smallest throughput of~$G$ for any (doubly stochastic) demand matrix~$\mathcal{M}$.

Keslassy, Chang, McKeown, and Lee showed the following result on the optimal demand-oblivious throughput.\footnote{Recently, similar results have also been obtained for dynamic settings~\cite{addanki2023mars,amir2022optimal,KCKL05,dc-throughput}.}
They provided a certain graph for each number~$n$ of nodes that achieves a throughput of~$\frac{n}{2n-1}$ and showed that no other graph has  throughput at least~$\frac{n}{2n-1}$.

\begin{theorem}[{\cite[Theorem 9]{KCKL05}}]
    \label{thm:oblivous}
    For each positive integer~$n$, there exists a directed graph~$G_n$ with~$n$ vertices such that for any~$n \times n$ demand matrix~$\mathcal{M}$, $G_n$ achieves throughput~$\frac{n}{2n-1}$ for~$\mathcal{M}$.
    For any graph~$G$ with~$n$ vertices and any~$\varepsilon > 0$, there exists a matrix~$\mathcal{M}_{G,\varepsilon}$ such that~$G$ does not achieve throughput~$\frac{n}{2n-1}+\varepsilon$ for~$\mathcal{M}_{G,\varepsilon}$.
\end{theorem}

Our paper is motivated by two fundamental questions. First, can these bounds be improved in demand-aware networks, that is, if we are allowed to chose the graph after knowing the input matrix, and second, what is the computational complexity for computing a graph maximizing (weak) throughput. 
As it is common in the related literature in the real of demand-aware networks (but not in the realm of demand-oblivious throughput), we require that all arcs have the same capacity. This also means that there is a fixed number~$N$ of arcs entering and leaving any given node.
The (unique) optimal demand-oblivious construction turned out to have~${N = 2n-1}$ and to make a meaningful separation from the demand-oblivious case, we decided to also pick~$N=2n-1$ for the demand-aware case.
This is allows us to state the four computational problems we study in this work:

\smallskip

\defproblem{\strong}{A doubly stochastic demand matrix~$\mathcal{M}$ and a number~$0 \leq \kappa \leq 1$.}{Is there a directed \Np-regular graph~$G$ such that the \old{} of~$G$ with respect to~$\mathcal{M}$ is at least~$\kappa$?}

\defproblem{\thr}{A doubly stochastic demand matrix~$\mathcal{M}$ and a number~$0 \leq \kappa \leq 1$.}{Is there a directed \Np-regular graph~$G$ such that the \new{} of~$G$ with respect to~$\mathcal{M}$ is at least~$\kappa$?}

\defproblem{\dt}{A doubly stochastic demand matrix~$\mathcal{M}$ and a number~$0 \leq \kappa \leq 1$.}{Is there a directed \Np-regular graph~$G$ such that the \direct{} of~$G$ with respect to~$\mathcal{M}$ is at least~$\kappa$?}

\defproblem{\di}{A doubly stochastic demand matrix~$\mathcal{M}$ and a number~$0 \leq \kappa \leq 1$.}{Is there a directed \Np-regular graph~$G$ such that the \direct{} of~$G$ with respect to~$\mathcal{M}$ is at least~$\kappa$?}

\subsection{Our Contribution}

We contribute fundamental insights into the achievable throughput bounds of demand-aware networks
and the underlying computational complexities.
The main finding of our work is a separation result, solving the open problem whether demand-awareness 
can improve the throughput also in the worst case. Specifically (more formally stated and proven as \cref{thm:main} later):

\begin{theorem}
    \label{thm:mainshort}
    For each~$\varepsilon > 0$, there exists an integer~$n_{\varepsilon}$ such that for each~$n \geq n_{\varepsilon}$ and each~$n \times n$ demand matrix~$\mathcal{M}$, there exists a graph~$G$ that achieves a throughput of~$\frac{5}{8}-\varepsilon$ for~$\mathcal{M}$.
\end{theorem}

More systematically, we show lower and upper bounds for all four problem variants we consider.
The results are shown in \cref{tab:bounds}.
\begin{table}[t]
    \setlength{\tabcolsep}{10pt} 
    \renewcommand{\arraystretch}{1.5}
    \centering
    \caption{Lower and upper bounds for the different notions of demand-aware throughput we consider in this paper. An entry~$A$ / $B$ means a lower bound of~$A$ and an upper bound of~$B$. The variable~$\varepsilon$ stands for an arbitrarily small positive constant. The entries marked with~$\dagger$ only hold for~$n=2$ and the entries marked with~$\ddagger$ only hold for a sufficiently large number~$n$ of nodes.}
    \begin{tabular}{c|c c}
         demand aware & throughput & weak \\\hline
         general & $\frac{5}{8}-\varepsilon^{\ddagger}$ {\Large /} $\frac{5}{6}^{\dagger}$ & $\frac{7n-4}{8n-4}$ {\Large /} $\frac{8}{9}^{\dagger}$ \\
         direct & $\frac{n}{2n-1}$ {\Large /} $\frac{n}{2n-1}$ & $\frac{7n-4}{8n-4}$ {\Large /} $\frac{7n-3}{8n-4}$
    \end{tabular}
    \label{tab:bounds}
\end{table}%
To the best of our knowledge, weak throughput has only been considered implicitly in the literature before~\cite{jyothi2016measuring} and we are not aware of any formal definitions and guarantees.
Our bounds for \dt{} and \di{} are essentially tight, whereas there are still gaps for \strong{} and \thr.

We also study the computational complexity of all four variants (depicted in \cref{tab:cc}) and show that both direct versions can be computed in polynomial time.
\begin{table}[t]
    \setlength{\tabcolsep}{10pt}
    \renewcommand{\arraystretch}{1.5}
    \centering
    \caption{The computational complexity of the different demand-aware throughput notions we consider.}
    \begin{tabular}{c|c c}
         demand aware & throughput & weak \\\hline
         general & open & NP-hard\\
         direct & $O(n^3)$ & $O(n^4 \log(n))$
    \end{tabular}
    \label{tab:cc}
\end{table}%
Regarding the multi-hop setting, we show that the demand-aware weak throughput is NP-hard 
to compute and the complexity of computing the demand-aware throughput remains an interesting open problem.
We mention that while polynomial-time algorithms for the exact weak throughput are impossible assuming~\mbox{P $\neq$ NP}, the lower bound listed in \cref{tab:bounds} is constructive and a graph achieving this bounds can be computed in polynomial time.

\subsection{Our Methods}

We next give a high-level overview of how the different results are achieved.
We start with the algorithms.
The optimal algorithm for \dt{} is a simple greedy algorithm and the argument why it is optimal also gives the tight lower bound.
The optimal algorithm for \di{} is a reduction to maximum-cost flow.
The NP-hardness for \thr{} is shown via an intricate reduction from \textsc{Exact Cover by 3-Sets}.
On a very high level, there is a vertex for each element and a gadget for each set in the original instance.
For each set gadget, there are two optimal solutions locally and we ensure that one of them can be used at most~$\frac{N}{3}$ times, where~$N$ is the number of elements in the original instance.
This solution also allows all nodes corresponding to elements in the respective set to fulfill a tiny bit of additional demand.
The value~$\kappa$ is then chosen such that all such nodes need to fulfill this additional demand and hence any solution needs to pick~$\frac{N}{3}$ sets that together contain all~$N$ elements, that is, they form an exact cover.

We continue with the lower and upper bounds, but we mention that all upper bounds are shown by constructing explicit (relatively simple) counterexamples and hence there is not much to say about them.
The lower bound for \dt{} follows from the greedy algorithm mentioned above.
Our main result is the lower bound for \strong.
This result builds upon two results that were previously used in approximation algorithms.
The first is called \emph{dependent rounding} and gives a randomized algorithm for a certain matrix-rounding problem.
The second result allows to apply Chernoff-type bounds in a setting with variables that are not independent but dependent in a particular way which appears in the context of dependent rounding.
In a nutshell, we use half of the arcs in a demand-aware manner using dependent rounding to greedily satisfy a large fraction of the demands directly.
The other half of the arcs are placed in a demand-oblivious manner.
We then send a certain fraction of the demands originating in any given vertex directly and therein mostly use the demand-aware arcs but potentially also some of the capacity of any demand-oblivious arc.
The dependent rounding ensures that the expected amount for any such arc is somewhat small and using the probabilistic method and the probability analysis mentioned above, we are able to show that for a sufficiently large number of nodes, some rounding must exist where the sum of capacities of the demand-oblivious arcs entering or leaving any particular vertex is at least~$\frac{3n}{4}-\varepsilon$.
These capacities then allow us to send all remaining demands via paths of length two using an adaptation of the demand-oblivious routing scheme.
In particular, for each pair~$(v_i,v_j)$ of nodes, we send demand from~$v_i$ to~$v_j$ via any given node~$v_h$ which is proportional to~$\min(c(i,h),c(h,j))$, where~$c(a,b)$ is the remaining capacity of the demand oblivious arc from~$v_a$ to~$v_b$.

The last result is an lower bound for \di{} (which is also a lower bound for \thr).
Here, we first show that given any matrix where all entries are between 0 and 1 and each row and column sums to an integer, there is a way to round all entries to 0 or 1 such that the sum in each row and column remains the same and the average entry rounded up is at least as large as the average entry in the entire matrix.
Them we use this result for the lower bound as follows.
We first greedily place arcs whose full capacity can be used to serve demand directly.
Then, we use the above rounding to ensure that from the remaining demand, we can satisfy at least the average remaining demand times the number of remaining arcs to be placed.
We also show that this gives a lower bound of~$\frac{7n-4}{8n-4}$.

The rest of the paper is structured as follows.
We first show in \cref{sec:relations} how the different throughput notions relate to one another.
\cref{sec:strong} is devoted to showing our upper and lower bounds for \strong{} and \dt.
Our upper and lower bounds for \thr{} and \di{} are shown in \cref{sec:weak} and our algorithms and the NP-hardness result are proven in \cref{sec:cc}.
We conclude with \cref{sec:conclusion}.

\section{Relations Between Different Throughput Notions}
\label{sec:relations}

We first investigate how the different throughput notions relate to each other.
Our findings are summarized in \cref{fig:relations}.
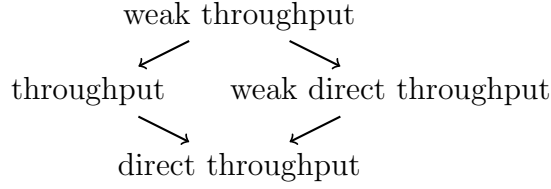
\begin{figure}
    \centering
    \begin{tikzpicture}
        \node at(0,2) (weak) {weak throughput};
        \node at(-2,1) (strong) {throughput};
        \node at(2,1) (weak direct) {weak direct throughput};
        \node at(0,0) (strong direct) {direct throughput};
        \draw[thick,->] (weak) to (strong);
        \draw[thick,->] (weak direct) to (strong direct);
        \draw[thick,->] (weak) to (weak direct);
        \draw[thick,->] (strong) to (strong direct);
    \end{tikzpicture}
    \caption{Relations of different demand-aware throughput notions. An arrow from a notion~$A$ to a notion~$B$ indicates that~$A \geq B$ for all demand matrices.}
    \label{fig:relations}
    \ifarxiv\else\Description[Hierarchy of throughput notions.]{Weak Throughput is the largest and direct throughput is the smallest. The other two (throughput and weak direct throughput) are incomparable in between the other two.}\fi
\end{figure}%
Note that any single arc also describes a non-empty path.
This simple observation yields the following.
\begin{observation}
    \label{obs:direct}
    Let~$\mathcal{M}$ be a doubly-stochastic~$n \times n$~matrix and~${\kappa \in [0,1]}$.
    If~$(\mathcal{M},\kappa)$ is a yes-instance of \dt, then it is also a yes-instance of \strong.
    Furthermore, if~$(\mathcal{M},\kappa)$ is a yes-instance of \di, then it is also a yes-instance of \thr.
\end{observation}
Next, we show that the weak (direct) throughput gives an upper bound for the (direct) throughput.
To this end, consider any instance~$(\mathcal{M},\kappa)$ of \strong{} or \dt, where~$\mathcal{M} = (a_{i,j})_{i,j \in [n]}$.

Let~${a'_{i,j} = \kappa a_{i,j}}$.
Since~${0 \leq \kappa \leq 1}$, it holds that~$a'_{i,j} \leq a_{i,j}$.
As each row of the resulting matrix~$\mathcal{M}'$ sums to~$\kappa$ and there are~$n$ rows, it holds that~$\frac{\sum_{i,j \in V \times V} a'_{i,j}}{n} = \frac{n\kappa}{n} = \kappa$, that is, $(\mathcal{M},\kappa)$ is a yes-instance of \thr{} or \di, respectively.
This yields the following.

\begin{observation}
    \label{obs:weak}
    Let~$\mathcal{M}$ be a doubly-stochastic~$n \times n$~matrix and~${\kappa \in [0,1]}$.
    If~$(\mathcal{M},\kappa)$ is a yes-instance of \textsc{Demand-Aware\linebreak Throughput}, then it is also a yes-instance of \thr.
    Furthermore, if~$(\mathcal{M},\kappa)$ is a yes-instance of \dt, then it is also a yes-instance of \di.
\end{observation}

Given the above results, one might now wonder whether there is also a similar connection between the demand-aware throughput and the demand-aware weak direct throughput.
We next show that this is not the case.
Consider the matrices
\[\mathcal{M}_1 = \begin{pmatrix}
    \frac{1}{2}&\frac{1}{2}\\
    \frac{1}{2}&\frac{1}{2}
\end{pmatrix} \text{ and } \mathcal{M}_2 = \begin{pmatrix}
\frac{9}{10}&\frac{1}{10}\\
\frac{1}{10}&\frac{9}{10}
\end{pmatrix}.\]
The optimal networks for~$\mathcal{M}_1$ and~$\mathcal{M}_2$ for the demand-aware throughput and the demand-aware weak direct throughput are shown in \cref{fig:example}.
\begin{figure*}[t]
  \begin{subfigure}{0.43\textwidth}
    \centering
    \begin{tikzpicture}
        \node[circle,draw] at (-4,3) (a1) {$1$};
        \node[circle,draw] at (-2,3) (b1) {$2$};
        \draw[loop above,->] (a1) to (a1);
        \draw[loop above,->] (b1) to (b1);
        \draw[->,bend left=10] (a1) to (b1);
        \draw[->,bend left=10] (b1) to (a1);
        \draw[->,bend left=30] (a1) to (b1);
        \draw[->,bend left=30] (b1) to (a1);
    \end{tikzpicture}
    \caption{The graph optimizing throughput for~$\mathcal{M}_1$. Note that for~${\kappa = \frac{8}{9}}$, we can use the capacity of~$\frac{2}{3}$ from vertex~$1$ to vertex~$2$ to send a flow of~$\frac{4}{9}$ directly and to send twice a flow of~$\frac{1}{9}$ from~$1$ to~$2$ (once to start a flow which is then sent back to vertex~$1$ and once to complete a symmetric flow starting in vertex~$2$). The flow from vertex~$1$ to itself is then~$\frac{1}{3} + \frac{1}{9} = \frac{4}{9}$ and all arc capacities are fully saturated.}
  \end{subfigure}
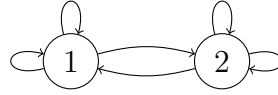
\hspace*{1cm}
  \begin{subfigure}{.43\textwidth}
    \centering
    \begin{tikzpicture}
        \node[circle,draw] at (-4,3) (a1) {$1$};
        \node[circle,draw] at (-2,3) (b1) {$2$};
        \draw[loop above,->] (a1) to (a1);
        \draw[loop left,->] (a1) to (a1);
        \draw[loop above,->] (b1) to (b1);
        \draw[loop right,->] (b1) to (b1);
        \draw[->,bend left=15] (a1) to (b1);
        \draw[->,bend left=15] (b1) to (a1);
    \end{tikzpicture}
    \caption{The graph optimizing throughput for~$\mathcal{M}_2$. Note that for~${\kappa = \frac{100}{114}}$, we can use the capacity of~$\frac{1}{3}$ from vertex~$1$ to vertex~$2$ to send a flow of~$\frac{10}{114}$ directly and to send twice a flow of~$\frac{14}{114}$ from~$1$ to~$2$ (once for flow from~$1$ to~$1$ and once for flow from~$2$ to~$2$). The flow from vertex~$1$ to itself is then~${\frac{2}{3} + \frac{14}{114} = \frac{2 \cdot 38+14}{114} = \frac{90}{114} = \frac{100}{114} \cdot \frac{9}{10}}$ and all arc capacities are fully saturated since~${\frac{10+2\cdot 14}{114} = \frac{38}{114} = \frac{1}{3}}$.}
  \end{subfigure}
  
  
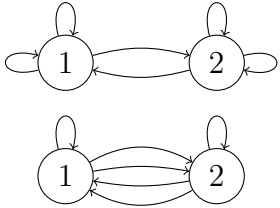
\begin{subfigure}{.43\textwidth}
    \centering
    \begin{tikzpicture}
        \node[circle,draw] at (-4,2.5) (a1) {$1$};
        \node[circle,draw] at (-2,2.5) (b1) {$2$};
        \draw[loop above,->] (a1) to (a1);
        \draw[loop left,->] (a1) to (a1);
        \draw[loop above,->] (b1) to (b1);
        \draw[loop right,->] (b1) to (b1);
        \draw[->,bend left=15] (a1) to (b1);
        \draw[->,bend left=15] (b1) to (a1);
        \node[circle,draw] at (-4,1) (a2) {$1$};
        \node[circle,draw] at (-2,1) (b2) {$2$};
        \draw[loop above,->] (a2) to (a2);
        \draw[loop above,->] (b2) to (b2);
        \draw[->,bend left=10] (a2) to (b2);
        \draw[->,bend left=10] (b2) to (a2);
        \draw[->,bend left=30] (a2) to (b2);
        \draw[->,bend left=30] (b2) to (a2);
    \end{tikzpicture}
    \caption{The graphs optimizing weak direct throughput for~$\mathcal{M}_1$. Note that in both cases each node can satisfy a flow of~$\frac{1}{2}$ to one of the two vertices and a flow of~$\frac{1}{3}$ to the other, that is~$\frac{5}{6}$ in total.}
    \end{subfigure}
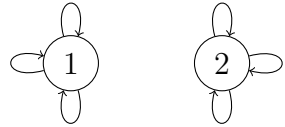
\hspace*{1cm}
  \begin{subfigure}{.43\textwidth}
    \centering
    \begin{tikzpicture}
        \node[circle,draw] at (-4,3) (a1) {$1$};
        \node[circle,draw] at (-2,3) (b1) {$2$};
        \draw[loop above,->] (a1) to (a1);
        \draw[loop left,->] (a1) to (a1);
        \draw[loop above,->] (b1) to (b1);
        \draw[loop right,->] (b1) to (b1);
        \draw[loop below,->] (a1) to (a1);
        \draw[loop below,->] (b1) to (b1);
    \end{tikzpicture}
    \caption{The graph optimizing weak direct throughput for~$\mathcal{M}_2$. Note that each node can satisfy~$0.9$ of the demand originating in it, that is, the weak direct throughput is (at least)~$\frac{9}{10}$.}
    \end{subfigure}
\caption{The directed $3$-regular graphs optimizing the throughput and the weak direct throughput for matrices $\mathcal{M}_1 = \begin{pmatrix}
    \frac{1}{2}&\frac{1}{2}\\
    \frac{1}{2}&\frac{1}{2}
\end{pmatrix}$ and $\mathcal{M}_2 = \begin{pmatrix}
\frac{9}{10}&\frac{1}{10}\\
\frac{1}{10}&\frac{9}{10}
\end{pmatrix}$.}
\label{fig:example}
\ifarxiv\else\Description[Directed 3-regular graphs optimizing different throughput requirements.]{The top left shows the directed 3-regular graph where each vertex has one self-loop and two arcs to the other vertex. The top right shows the graph with two self loops and one arc to the respective other vertex. The bottom left shows two graphs, the same as in the top left and the top right. The bottom right shows the graph where both vertices have each 3 self-loops.}\fi
\end{figure*}
On the one hand, it holds that~$(\mathcal{M}_1,\frac{8}{9})$ is a yes-instance of \strong{} and a no-instance of \di{} (since the highest weak direct throughput that is achievable for~$\mathcal{M}_1$ is~$\frac{5}{6}$).
On the other hand, $(\mathcal{M}_2,\frac{9}{10})$ is a yes-instance of \di, but a no-instance of \strong{} (as the highest achievable throughput for~$\mathcal{M}_2$ is~$\frac{100}{114} < \frac{9}{10}$).

To conclude this section, we observe that since the graph~$G$ behind \cref{thm:oblivous} is~$\Np$-regular, it holds that the demand-aware throughput for any demand matrix is at least as much as the one achieved by~$G$, that is, the demand-aware throughput is at least as good as the demand-oblivious throughput (where no restriction is made on the number of arcs in~$G$ other than that the sum of capacities outgoing and incoming to each node is~$1$).

\begin{observation}
    For any doubly-stochastic $n \times n$ demand matrix~$\mathcal{M}$ and any~$\kappa \leq \frac{n}{2n-1}$, it holds that~$(\mathcal{M},\kappa)$ is a yes-instance of \strong.
\end{observation}

\section{Throughput Bounds}
\label{sec:strong}
In this section, we show lower and upper bounds for the demand-aware (direct) throughput.
We start with a lower bound of~$\frac{5}{8}-\varepsilon$ for throughput for any~$\varepsilon > 0$ that relies on the probabilistic method.
Therein, we rely on randomized techniques previously applied in the context of approximation algorithms.
However, our final result will be a deterministic lower bound.
Afterwards, we investigate demand-aware direct throughput and show that a direct throughput of~$\frac{n}{2n-1}$ is always achievable and this bound is tight, that is, for any~$\varepsilon>0$, there exist matrices such that a throughput of~$\frac{n}{2n-1}+\varepsilon$ is unobtainable.

\subsection{Throughput}

We will next prove a lower and an upper bound for the demand-aware throughput of networks.
We start with the lower bound.
Therein, we make use the following result due to Gandhi et al.~\cite[Theorem 2.3]{GKPS06}.\footnote{We slightly reformulated the statement to fit our notation. The authors of the original result interpret any $n \times m$ matrix as the adjacency matrix of a biparitite graph with~$n$ vertices on one side and~$m$ vertices on the other side. We only require the special case where~$n = m$ and statement (P3) of the original paper only for the case~$b=0$.}

\begin{theorem}
    \label{thm:dependentrounding}
    Given a matrix~$\mathcal{M}=(\alpha_{i,j})_{i,j \in [n]}$ where each entry is between~$0$ and~$1$ and each row and each column sums to an integer value, there is a rounding scheme creating a matrix~$\mathcal{N} = (\beta_{i,j})_{i,j \in [n]}$ where~$\beta_{i,j} \in \{0,1\}$ and
    \begin{enumerate}
        \item $\beta_{i,j} = 1$ with probability (exactly)~$\alpha_{i,j}$,
        \item the sum of entries in~$\mathcal{N}$ in any row or column is the same as the corresponding sum in~$\mathcal{M}$ (with probability 1),
        \item for all~$i \in [n]$ and for any subset~$S \subseteq [n]$, the probability that~$\beta_{i,j} = 0$ for all~$j \in S$ is at most~$\prod_{j \in S} (1-\alpha_{i,j})$, and
        \item for all~$j \in [n]$ and for any subset~$S \subseteq [n]$, the probability that~$\beta_{i,j} = 0$ for all~$i \in S$ is at most~$\prod_{i \in S} (1-\alpha_{i,j})$.
    \end{enumerate}
\end{theorem}

We also use a result due to Panconesi and Srinivasan~\cite{PS97} (in a reformulation due to Gandhi et al.~\cite[Theorem 3.1]{GKPS06}).
Intuitively, this can be thought of as a guarantee that the rounding from the previous result does not deviate from the expected value by more than an~$\varepsilon$-fraction in any row or column.
A little more formally, it bounds the probability that the sum of entries in a given row or column is at most~$(1-\delta)\mu$ for any given~$\delta > 0$, where~$\mu$ is the expected sum in that row or column.

\ifarxiv
\begin{theorem}
    \label{thm:tailprob}
    Let~$0 \leq a_1,a_2,\ldots,a_t,x_1,x_2,\ldots,x_t \leq 1$, let~$X_1,X_2,\ldots,X_t \in \{0,1\}$ be random variables, where~$X_i = 1$ has probability~$x_i$, and let~${\mu = \sum_{i \in [t]}a_ix_i}$ be the expected value of~$\sum_{i \in [t]} a_i X_i$.
    Suppose for all~$S \subseteq [t]$, it holds that the probability that~$X_i = 0$ for all~$i \in S$ is at most~$\prod_{i \in S}(1-x_i)$.
    Then, for any~$\delta \in [0,1]$, it holds that the probability that~$\sum_{i \in [t]} a_i X_i \leq \mu(1-\delta)$ is at most~$e^{-\frac{\mu\delta^2}{2}}$.
\end{theorem}
\else
\begin{theorem}
    \label{thm:tailprob}
    Let~$a_1,a_2,\ldots,a_t,x_1,x_2,\ldots,x_t$ be values between~$0$ and~$1$, let~$X_1,X_2,\ldots,X_t$ be random variables taking values~$0$ or~$1$, where~$X_i = 1$ has probability~$x_i$, and let~${\mu = \sum_{i \in [t]}a_ix_i}$ be the expected value of~$\sum_{i \in [t]} a_i X_i$.
    Suppose for all~$S \subseteq [t]$, it holds that the probability that~$X_i = 0$ for all~$i \in S$ is at most~$\prod_{i \in S}(1-x_i)$.
    Then, for any~$\delta \in [0,1]$, it holds that the probability that~$\sum_{i \in [t]} a_i X_i \leq \mu(1-\delta)$ is at most~$e^{-\frac{\mu\delta^2}{2}}$.
\end{theorem}
\fi

We can now show our main result.

\begin{theorem}
    \label{thm:main}
    For each~$\kappa < \frac{5}{8}$, there exists an integer~$n_{\kappa}$ such that for all~$n \geq n_{\kappa}$ and each doubly stochastic~$n \times n$ matrix~$\mathcal{M}$, the instance~$(\mathcal{M},\kappa)$ is a yes-instance of \strong.
\end{theorem}

\begin{proof}
Let~$\varepsilon = \kappa - \frac{5}{8} > 0$.
We assume that~$\varepsilon < \frac{1}{8} - \frac{1}{4n-2}$ as otherwise~${\kappa \leq \frac{1}{2} + \frac{1}{4n-2} = \frac{n}{2n-1}}$ and \cref{thm:oblivous} shows that~$(\mathcal{M},\kappa)$ is a yes-instance of \strong.
Moreover, it holds that~${(2n-1)\kappa > n-1}$.
Let~$n_{\kappa}$ be the smallest positive integer such that~${n_{\kappa} \geq \frac{6}{16\varepsilon}+1}$ and~$e^{-\frac{2\varepsilon^2n_{\kappa}}{3}} < \frac{1}{2n_{\kappa}}$.
Note that such a number always exists as~$e^{-xn} < \frac{1}{yn}$ holds for any pair~$x,y > 0$ for all sufficiently large~$n$.
Let~$\mathcal{M} = (a_{i,j})_{i,j \in [n]}$ be any doubly stochastic ${n \times n}$~matrix with~$n \geq n_{\kappa}$.
We will show that~$(\mathcal{M},\kappa)$ is a yes-instance of \strong.

Let~$b_{i,j} = (n-1) a_{i,j}$ and~$b'_{i,j} = b_{i,j} - \lfloor b_{i,j} \rfloor$ for all~${i,j \in [n]}$.
Recall that~${(2n-1)\kappa > n-1}$.
Hence, ${b_{i,j} < (2n-1)\kappa a_{i,j}}$.
We will employ a randomized rounding technique (based on \cref{thm:dependentrounding}) and fix the rounding later.
For now, let~$X_{i,j}$ be a random variable that takes values in~$\{0,1\}$ where~$X_{i,j}=1$ corresponds to rounding entry~$b_{i,j}$ up and otherwise~$b_{i,j}$ is rounded down.
The probability that~$X_{i,j} = 1$ is~$b'_{i,j}$.
Now, consider any rounding~$\varphi$ of the matrix~${\mathcal{N}'=(b'_{i,j})_{i,j \in [n]}}$ that \cref{thm:dependentrounding} could return and let~${c_{i,j} \in \{\lfloor b_{i,j} \rfloor,\lceil b_{i,j} \rceil\}}$ be the result of applying the same rounding to~$b_{i,j}$.
Let~$d_{i,j} = c_{i,j}+1 > b_{i,j}$.
Note that due to property (2) in \cref{thm:dependentrounding}, each row~$i$ in the matrix~$\mathcal{N}=(d_{i,j})_{i,j\in[n]}$ sums to~$\sum_{j=1}^n (b_{i,j} + 1) = (n-1) + n = 2n-1$ and the same holds for each column~$j$.

We will now construct the graph that can host~$\kappa \mathcal{M}$ and for notational convenience, we will consider the names of the vertices in it~$1,2,\ldots,n$.
For a given rounding~$\varphi$ and any~$i,j \in [n]$, we add~$d_{i,j}$ arcs from vertex~$i$ to vertex~$j$.
Note that this can be seen as adding~$c_{i,j}$ demand-aware arcs and one additional demand-oblivious arc between each pair of (not necessarily distinct) vertices.
We partition the flow into two parts and describe them in two separate stages.
In the first stage, we send~$b_{i,j} < \min(a_{i,j},d_{i,j})$ demand from vertex~$i$ to vertex~$j$ directly, that is, we add the pair~$(((i,j)),b_{i,j})$ to the solution collection~$\mathcal{P}$.

For the second stage, we introduce some additional notation.
\ifarxiv
See \cref{fig:prob} for an example instance with all of the following definitions. 
\begin{figure*}
    \centering
    \begin{tikzpicture}
    \renewcommand{\arraystretch}{1.2}
        \node at(-3,3) {$\mathcal{M} = (a_{i,j})_{i,j \in [3]} = \begin{pmatrix}
            \frac{3}{8} & \frac{1}{4} & \frac{3}{8}\\
            \frac{3}{8} & \frac{3}{8} & \frac{1}{4}\\
            \frac{1}{4} & \frac{3}{8} & \frac{3}{8}
        \end{pmatrix}$};
        \node at(3,3) {$(n-1)\mathcal{M} = (b_{i,j})_{i,j \in [3]} = \begin{pmatrix}
            \frac{3}{4} & \frac{1}{2} & \frac{3}{4}\\
            \frac{3}{4} & \frac{3}{4} & \frac{1}{2}\\
            \frac{1}{2} & \frac{3}{4} & \frac{3}{4}
        \end{pmatrix}$};
        \node at(-3,1) {$(2n-1)\kappa \mathcal{M} = \begin{pmatrix}
            \frac{75}{64} & \frac{25}{32} & \frac{75}{64}\\
            \frac{75}{64} & \frac{75}{64} & \frac{25}{32}\\
            \frac{25}{32} & \frac{75}{64} & \frac{75}{64}
        \end{pmatrix}$};
        \node at(3,1) {rounding~$(c_{i,j})_{i,j \in [3]} =  \begin{pmatrix}
            1 & 0 & 1\\
            1 & 1 & 0\\
            0 & 1 & 1
        \end{pmatrix}$};
        \node at(-3,-1) {$(d_{i,j})_{i,j \in [3]} =  \begin{pmatrix}
            2 & 1 & 2\\
            2 & 2 & 1\\
            1 & 2 & 2
        \end{pmatrix}$};
        \node at(-4.25,-1.5) {$= (c_{i,j}+1)_{i,j \in [3]}$};
        \node at(-3.5,-3) {excess capacity~$(\eta_{i,j})_{i,j \in [3]} =  \begin{pmatrix}
            1 & \frac{1}{2} & 1\\
            1 & 1 & \frac{1}{2}\\
            \frac{1}{2} & 1 & 1
        \end{pmatrix}$};
        \node at(-4.2,-3.75) {$= \Big(\min(d_{i,j} - b_{i,j},1)\Big)_{i,j \in [3]}$};
        \node at(3.75,-1) {overflow~$(\sigma_{i,j})_{i,j \in [3]} =  \begin{pmatrix}
            \frac{27}{64} & \frac{9}{32} & \frac{27}{64}\\
            \frac{27}{64} & \frac{27}{64} & \frac{9}{32}\\
            \frac{9}{32} & \frac{27}{64} & \frac{27}{64}
        \end{pmatrix}$};
        \node at(2.25,-1.75) {$=\Big((2n-1)\kappa a_{i,j}-b_{i,j}\Big)_{i,j \in [3]}$};
        \node at(-4.25,-1.5) {$= (c_{i,j}+1)_{i,j \in [3]}$};
        \node at(3.5,-3) {common excess~$(\zeta_{i,j})_{i,j \in [3]} =  \begin{pmatrix}
            2 & 2 & \frac{5}{2}\\
            \frac{5}{2} & 2 & 2\\
            2 & \frac{5}{2} & 2
        \end{pmatrix}$};
        \node at(2.4,-3.75) {$= \Big(\sum_{h \in [n]}\min(\eta_{i,h},\eta_{h,j})\Big)_{i,j \in [3]}$};
    \end{tikzpicture}
    \caption{An example of our approach for the second stage for~$\kappa = \frac{5}{8}$ (that is,~$\varepsilon=0$) and~$n = 3$. We mention that the approach does not work as~$\varepsilon$ does not satisfy the two conditions we have for it. We still chose~$\varepsilon=0$ for the sake of clarity of presentation of the example.}
    \label{fig:prob}
\end{figure*}
\fi
Let~$\eta_{i,j} = \min(d_{i,j} - b_{i,j},1)$ be the \emph{excess capacity} from vertex~$i$ to vertex~$j$.
Note that at least~$\eta_{i,j}$ capacity is not used so far on the arcs going from vertex~$i$ to vertex~$j$.
The expected excess capacity is~${b'_{i,j} \cdot 1 + (1-b'_{i,j}) \cdot (1-b'_{i,j}) = (b'_{i,j})^2 -b'_{i,j} + 1 \geq \frac{3}{4}}$.
Here, the last inequality holds because the derivative of~${f(x)=x^2-x+1}$ is~$2x-1$, which evaluates to~$0$ at~$x=\frac{1}{2}$, which is also the minimum entry for~$f$ and~$f(\frac{1}{2}) = \frac{3}{4}$.
By the linearity of the expectation, the expected sum of excess in a single row or column is at least~$\frac{3n}{4}$.

We will next prove that the probability that the excess capacity in any fixed row or column is at most~$(\frac{3}{4}-\varepsilon)n$ is smaller than~$\frac{1}{2n}$.
Since the argument is completely symmetric for rows and columns, we will only show it for a fixed row.
To this end, fix any row~$i$ and and let~$Y_i = \sum_{j \in [n]} (1-b'_{i,j})$.
Note that the sum of excess capacities in row~$i$ is at least~$Y_i$.
More precisely, it is~$Y_i + \sum_{j \in [n]} X_{i,j} b'_{i,j}$.
Applying \cref{thm:tailprob} with~$t = n + Y_i$, $a_j = b'_{i,j} = x_j$ for all~$j \in [n]$, $a_j = 1 = x_j$ for all~$j \in \{n+1,n+2,\ldots,n+Y_i\}$, and~$\delta = \frac{4\varepsilon}{3}$ yields that the probability that the excess capacity in row~$i$ is at most~$\frac{3n}{4}(1-\delta)$ is at most~$e^{-\frac{\frac{3n}{4}\delta^2}{2}}$.
Note that~$\frac{3n}{4}(1-\delta) = \frac{3n}{4} - \frac{3}{4}\delta n = (\frac{3}{4}-\varepsilon)n$.
Thus, the probability that the excess capacity in row~$i$ is at most~$(\frac{3}{4}-\varepsilon)n$ is at most~$e^{-\frac{\frac{3n}{4}\delta^2}{2}} = e^{-\frac{3}{8}n (\frac{4}{3}\varepsilon)^2} = e^{-\frac{2}{3}n\varepsilon^2} < \frac{1}{2n}$.

We are now in a position to fix the specific rounding.
To this end, we say that a rounding is good for a certain row or column, if the sum of excess capacities in that row/column is at least~$(\frac{3}{4}-\varepsilon)n$.
The probability that a random rounding is good for a given row or column is larger than~$1-\frac{1}{2n}$ as proven above.
The probability that any given rounding is good for~$x$ fixed rows and/or columns is larger than~$1 - \frac{x}{2n}$.
This is true since even in the worst case bad events can be at most disjoint and the likelihood of at least one of them happening is at most the sum of their individual probabilities.
Hence, the probability that a rounding is good for all~$n$ rows and all~$n$ columns is larger than~$1 - \frac{2n}{2n} = 0$.
Thus, the probability is strictly larger than~$0$ and thus such a rounding necessarily exists (but we do not know how to compute it).
We consider~$\varphi$ to be such a rounding for the rest of the proof.

Now that we fixed the rounding, it remains to prove that the constructed graph can host the matrix~$\kappa\mathcal{M}$.
To this end, we already showed that~$b_{i,j}$ flow can be sent from vertex~$i$ to vertex~$j$ directly without using any excess capacity.
We next show that all remaining flow can be satisfied using only the excess capacities and paths of length two.
To this end, we first define the \emph{common excess of row~$i$ and column~$j$} as~$\zeta_{i,j} = \sum_{h \in [n]}\min(\eta_{i,h},\eta_{h,j})$.
Since~${0 \leq \eta_{a,b} \leq 1}$ for all~$a,b \in [n]$ and the excess capacity of each row and column is at least~$(\frac{3}{4}-\varepsilon)n$, it holds that~$\zeta_{i,j} \geq (\frac{1}{2}-2\varepsilon)n$.
This is true since all~$\eta_{i,h}$ and~$\eta_{h,j}$ being~$1$ results in~$\zeta_{i,j}=n$ and reducing any entry by an amount~$\delta$ can reduce~$\zeta_{i,j}$ by at most~$\delta$.
Since the total reduction in row~$i$ is at most~$(\frac{1}{4}+\varepsilon)n$ and the total reduction in column~$j$ is at most~$(\frac{1}{4}+\varepsilon)n$, it holds that~$\zeta_{i,j} \geq n - 2(\frac{1}{4}+\varepsilon)n = (\frac{1}{2}-2\varepsilon) n$.

Next, let~${\sigma_{i,j}=((2n-1)\kappa-(n-1))a_{i,j}}$ for each~$i,j \in [n]$ be the \emph{overflow from~$i$ to~$j$}.
Note that we already saturated~$b_{i,j}=(n-1)a_{i,j}$ flow directly and hence the overflow describes exactly the amount of flow that still needs to be scheduled.
For each~$i,j,h \in [n]$, we send~$\sigma_{i,j} \cdot \frac{\min(\eta_{i,h},\eta_{h,j})}{\zeta_{i,j}}$ flow over arcs~$(i,h)$ and~$(h,j)$, that is, we add the pairs
\[(P=((i,h),(h,j)),\sigma_{i,j} \cdot \frac{\min(\eta_{i,h},\eta_{h,j})}{\zeta_{i,j}})\] to the solution collection~$\mathcal{P}$ of flows.
To conclude the proof, we will show that~$\sigma_{i,j}$ flow is sent from~$i$ to~$j$ in this way and that from any vertex~$a$ to any vertex~$b$, at most the excess capacity for this pair is used for all of these flows combined.

For the first point, note that~$\zeta_{i,j} = \sum_{h \in [n]}\min(\eta_{i,h},\eta_{h,j})$ by definition.
Thus, \[\sum_{h \in [n]} \sigma_{i,j} \cdot \frac{\min(\eta_{i,h},\eta_{h,j})}{\zeta_{i,j}} = \frac{\sigma_{i,j}}{\zeta_{i,j}} \sum_{h \in [n]} \min(\eta_{i,h},\eta_{h,j}) =  \frac{\sigma_{i,j} \zeta_{i,j}}{\zeta_{i,j}} = \sigma_{i,j}.\]
This shows that all overflow is sent from~$i$ to~$j$ in this manner.

For the second point, recall that~$\zeta_{i,j} \geq (\frac{1}{2}-2\varepsilon)n$ for all~$i,j \in [n]$.
Moreover, note that the sum of overflows originating from~$i$ is
\begin{align*}
    \sum_{j \in [n]} \sigma_{i,j} &= \sum_{j \in [n]} ((2n-1)\kappa-(n-1))a_{i,j} \\
    &= (2n-1)\kappa-(n-1) \sum_{j \in [n]} a_{i,j} = (2n-1)\kappa-(n-1).
\end{align*}
\ifarxiv Similarly,\else By the same argument as above, \fi the sum of overflows terminating in~$j$ is~${\sum_{i\in [n]}\sigma_{i,j} = (2n-1)\kappa-(n-1)}$.
Any arc~$(a,b)$ is only used for two types of flows in the second step: overflow terminating in vertex~$b$ and overflow originating from vertex~$a$.
Thus, the maximum flow sent over the arc(s)~$(a,b)$ in the second step is at most
\begin{align*}
    &\quad\sum_{i \in [n]}(\sigma_{i,b} \frac{\min(\eta_{i,a},\eta_{a,b})}{\zeta_{i,b}}) + \sum_{j\in [n]}(\sigma_{a,j} \frac{\min(\eta_{a,b},\eta_{b,j})}{\zeta_{a,j}})\\
    &\leq \sum_{i \in [n]}(\sigma_{i,b} \frac{\eta_{a,b}}{(\frac{1}{2}-2\varepsilon)n}) + \sum_{j\in [n]}(\sigma_{a,j} \frac{\eta_{a,b}}{(\frac{1}{2}-2\varepsilon)n})\\
\ifarxiv
    &\leq \eta_{a,b} \frac{2((2n-1)\kappa-(n-1))}{(\frac{1}{2}-2\varepsilon)n}\\
    &\leq \eta_{a,b} \frac{2((2n-1)(\frac{5}{8}-\varepsilon)-(n-1))}{(\frac{1}{2}-2\varepsilon)n}\\
\fi
    &\leq \eta_{a,b}\frac{2(\frac{10n}{8}-2n\varepsilon-\frac{5}{8}+\varepsilon-(n-1))}{(\frac{1}{2}-2\varepsilon)n}\\
    &\leq \eta_{a,b} \frac{\frac{n}{2}-4n\varepsilon+\frac{6}{8}+2\varepsilon}{(\frac{1}{2}-2\varepsilon)n}\\
    &\leq \eta_{a,b} \frac{(\frac{1}{2}-2\varepsilon)n -2n\varepsilon + \frac{3}{4}+2\varepsilon}{(\frac{1}{2}-2\varepsilon)n}\\
    &\leq \eta_{a,b} (1-\frac{2n\varepsilon - \frac{3}{4} - 2\varepsilon}{(\frac{1}{2}-2\varepsilon)n}) \leq \eta_{a,b}. 
\end{align*}
Here, the last inequality follows from the fact that~$n > 0$, $2\varepsilon < \frac{1}{2}$ because~$\varepsilon < \frac{1}{8} < \frac{1}{4}$, and~$2n\varepsilon \geq 2\varepsilon+\frac{3}{4}$ because~$n \geq \frac{6}{16\varepsilon}+1$.
Thus, the total flow sent from~$a$ to~$b$ directly in the second step is at most the excess capacity~$\eta_{a,b}$.
This concludes the proof.
\end{proof}

While there is no formula in closed form for the number of vertices needed to guarantee a certain throughput, we can calculate a few entries.
The number of vertices needed to guarantee that our approach improves upon the demand-oblivious bound of~$\frac{n}{2n-1}$ is roughly~$700$ (with a value of~$\kappa=\frac{5}{8}-\frac{1}{8.03}$).
The limit of our approach is~$\kappa = \frac{5}{8} = 0.625$.
For~$n \geq 2250$, we can guarantee a throughput of~$\kappa = 0.55$, for~$n \geq 26000$, we can guarantee a throughput of~$\kappa=0.6$, and for a guarantees throughput of~$\kappa=0.624$, roughly 27 million vertices are required.
However, we mention that our goal was to show a separation, that is, \emph{any} improvement over the oblivious bound and we conjecture that our approach can be improved to give better bounds for much fewer vertices.
We conclude this subsection with an upper bound of~$\frac{5}{6}$ for the throughput.

\begin{proposition}
    For each~$\kappa > \frac{5}{6}$, there exists a doubly stochastic~${2 \times 2}$~demand matrix~$\mathcal{M}_{\kappa}$ such that~$(\mathcal{M}_{\kappa},\kappa)$ is a no-instance of \strong.
\end{proposition}

\begin{proof}
    Let~$\varepsilon = \kappa - \frac{5}{6} \leq \frac{1}{6}$.
    We construct~$\mathcal{M}_{\kappa} = \begin{pmatrix}
        1-\varepsilon & \varepsilon\\\varepsilon & 1-\varepsilon
    \end{pmatrix}$.
    There are only four directed $3$-regular graphs with two vertices, which are depicted in \cref{fig:fourgraphs}.
        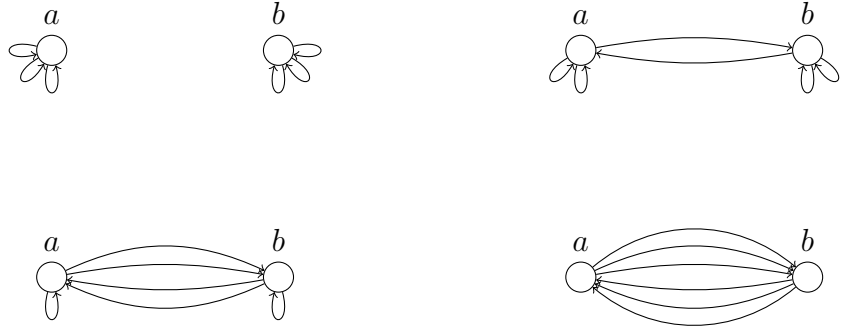
\begin{figure*}[t]
        \centering
        \begin{tikzpicture}
            \def\x{3}
            \def\y{4}
            \foreach \i in {0,1}{
                \node[circle,draw,inner sep=4pt,label=$a$] at(\x*\i+\y*\i,0) (a\i) {};
                \node[circle,draw,inner sep=4pt,label=$b$] at(\x*\i+\y*\i+\x,0) (b\i) {};
            }
            \foreach \i in {2,3}{
                \node[circle,draw,inner sep=4pt,label=$a$] at(\x*\i+\y*\i-2*\x-2*\y,-3) (a\i) {};
                \node[circle,draw,inner sep=4pt,label=$b$] at(\x*\i+\y*\i-\x-2*\y,-3) (b\i) {};
            }
            \foreach \i in {0,1,2,3}{
                \foreach \j in {0,1,2,3}{
                    \ifthenelse{\j < \i}{
                        \draw[->,bend left=15*\j+10] (a\i) to (b\i);
                        \draw[->,bend left=15*\j+10] (b\i) to (a\i);
                    }{
                        \ifthenelse{\i = \j}{}{
                        \draw[->,in=45*\j+150,out=45*\j+120,loop] (a\i) to (a\i);
                        \draw[->,in=-45*\j+30,out=-45*\j+60,loop] (b\i) to (b\i);
                        }
                    }
                }
            }
        \end{tikzpicture}
        \caption{The four different vertex-labeled directed $3$-regular graphs with two vertices. Note that each vertex has exactly~$3$ outgoing arcs and out of these~$c \in \{0,1,2,3\}$ are self-loops and the other~$3-c$ are going to the other vertex.}
        \label{fig:fourgraphs}
        \ifarxiv\else\Description[The four different directed 3-regular graphs.]{The top left shows a graph with two vertices that each have three self-loops. The top tight shows a graph with two vertices each having two self-loops and one arc to the respective other vertex. The bottom left shows a graph with two vertices, each with one self-loop and two arcs to the respective arc. The bottom right shows a graph with two vertices and six arcs between them, three in each direction.}\fi
    \end{figure*}
    The graph in the top left cannot host any demand between the two nodes (and~$\varepsilon > 0$).
    The two graph on the bottom cannot host~$\frac{5}{6} \mathcal{M}_{\kappa}$ since~$1-\varepsilon > \frac{2}{3}$.
    Thus, the only possible graph to host~$\kappa \mathcal{M}_{\kappa}$ is the graph in the top right.
    This graph can send~$\kappa\varepsilon$ directly from each node to the other, $\frac{2}{3}$ directly via the self loops, and for at least one of the two nodes at most~$\frac{\frac{1}{3}-\kappa\varepsilon}{2}$ indirectly from the node to itself via the other node.
    Thus, the total demand satisfied from this node to itself is at most
    \begin{align*}
    \frac{2}{3} + \frac{1}{6} - \frac{\kappa\varepsilon}{2} &= \frac{5}{6} - \frac{5 \varepsilon}{12} - \frac{\varepsilon^2}{2} = \frac{5}{6}  - \frac{5\varepsilon}{6} + \varepsilon - \varepsilon - \frac{\varepsilon^2}{2} \\&< \frac{5}{6} + \frac{\varepsilon}{6} - \varepsilon^2 = (\frac{5}{6}+\varepsilon)(1-\varepsilon) = \kappa(1-\varepsilon).\end{align*}
    Here, the inequality is due to the fact that~$\varepsilon \ifarxiv > \varepsilon^2 \fi > \frac{\varepsilon^2}{2}$ for all~${0 < \varepsilon < 1}$.
    This show that the graph in the top right cannot host~$\kappa(1-\varepsilon)$ both from~$a$ to~$a$ and from~$b$ to~$b$.
    Thus, this graph cannot host~$\kappa \mathcal{M}_{\kappa}$, concluding the proof.
\end{proof}

\subsection{Direct Throughput}
In this section, we show tight lower and upper bounds for the demand-aware direct throughput.
We start with a lower bound of~$\frac{n}{2n-1}$.
Afterwards, we show that this bound is tight, that is, for any~$\varepsilon > 0$ and any~$n$, there exists a doubly stochastic~$n \times n$ demand matrix~$\mathcal{M}$ such that no directed~$\Np$-regular graph can host~$(\frac{n}{2n-1} + \varepsilon)\mathcal{M}$ directly.

\begin{proposition}
    \label{prop:directstronglower}
    For each~$n \geq 1$, each~$\kappa \leq \frac{n}{2n-1}$, and each doubly stochastic~$n \times n$ matrix~$\mathcal{M}$, it holds that~$(\mathcal{M},\kappa)$ is a yes-instance of \dt.
\end{proposition}

\begin{proof}
    Let~$\mathcal{M}' = (2n-1)\mathcal{M} = (b_{i,j})_{i,j \in [n]}$.
    We construct a graph that can host~${\kappa \mathcal{M}}$ directly by describing its adjacency matrix~${\mathcal{A} = (a_{i,j})_{i,j \in [n]}}$.
    We start with the empty graph, that is, $\mathcal{A}$ is the~$n \times n$ matrix with all entries being~$0$ and increase the values until each row and column of~$\mathcal{A}$ sums to exactly~$2n-1$.
    We say that such a row or column is \emph{full} and denote the set of full rows by~$\mathcal{I}$ and the set of full columns by~$\mathcal{J}$.
    We compute~$\frac{\mathcal{M}'}{\mathcal{A}} = (c_{i,j}=\frac{b_{i,j}}{a_{i,j}})_{i,j \in [n]}$.
    If~$a_{i,j}=0$, then we define~$c_{i,j}=0$ if~$b_{i,j}=0$ and~$c_{i,j} = \infty$ if~$b_{i,j} > 0$.
    Let~$i,j$ be the indices among~$i \notin \mathcal{I}$ and~$j \notin \mathcal{J}$ where~$c_{i,j}$ is maximum.
    We then increase~$a_{i,j}$ by $1$, recompute~$c_{i,j}$, check whether~$i$ and/or~$j$ should be added to~$\mathcal{I}$ or~$\mathcal{J}$, respectively, and iteratively find the next entry to increase using the same strategy.

    Since in the end, each row and column is full, the constructed graph is a directed~$\Np$-regular graph.
    We next show that it can host~$\kappa \mathcal{M}$ directly.
    To this end, it suffices to show that~$c_{i,j} \leq \frac{2n-1}{n}$ whenever~$i \in \mathcal{I}$ or~$j \in \mathcal{J}$.
    Assume towards a contradiction that at some point there exists an entry~$c_{i^*,j^*}$ with~$i^* \in \mathcal{I}$ or~$j^* \in \mathcal{J}$ where~$c_{i^*,j^*} > \frac{2n-1}{n}$.
    Since the argument will be symmetric, we assume without loss of generality that~$i^* \in \mathcal{I}$.
    Now consider the moment in the algorithm when we increased the value~$a_{i^*,j'}$ for any~$j'$ for the last time and thereby making row~$i$ full.
    Define the \emph{mass}~$m_{i^*,j}$ for each column~$j$ as~$a_{i^*,j}-1$ at the considered moment.
    Let~$S$ be the set of all columns~$j$, where~$m_{i^*,j} \geq 1$ except for potentially column~$j^*$.
    Note that it holds for all~$j \in S$ that~${\frac{b_{i^*,j}}{m_{i^*,j}} > \frac{2n-1}{n}}$ as otherwise, whenever we increased~$a_{i^*,j}$ for the last time, we would have increased~$a_{i^*,j^*}$ instead by construction.
    Equivalently it holds for each~$j \in S$ that~$b_{i^*,j} > m_{i^*,j}\frac{2n-1}{n}$.
    Note that the sum of all masses in row~$i^*$ is at least~$(2n-1)-n = n-1$.
    Hence, it holds that~$\sum_{j \in S} m_{i^*,j} \geq n-1-m_{i^*,j^*}$.
    Moreover, recall that~$c_{i^*,j^*} > \frac{2n-1}{n}$, which is equivalent to~$b_{i^*,j^*} > (m_{i^*,j^*}+1) \frac{2n-1}{n}$.
    This implies
    \begin{align*}
        2n-1 &= \sum_{j \in [n]} b_{i^*,j} \geq \sum_{j \in S \cup \{j^*\}} b_{i^*,j} \\
        &= (\sum_{j \in S} b_{i^*,j}) + b_{i^*,j^*} \\
        &> (\sum_{j \in S} m_{i^*,j} \frac{2n-1}{n}) + (m_{i^*,j^*}+1) \frac{2n-1}{n}\\
        &\geq (n-1-m_{i^*,j^*} + m_{i^*,j^*} + 1) \frac{2n-1}{n} = 2n-1,
    \end{align*}
    a contradiction.
    Thus, each~$c_{i,j}$ with~$i \in \mathcal{I}$ or~$j \in \mathcal{J}$ is at most~$\frac{2n-1}{1}$.
    In the end, all rows and columns are full.
    Hence, $\frac{b_{i,j}}{a_{i,j}} = c_{i,j} \leq \frac{2n-1}{n}$ for all~${i,j \in [n]}$.
    That is, the graph can directly host~$\frac{n}{2n-1}\mathcal{M} \geq \kappa \mathcal{M}$.
    This concludes the proof.
\end{proof}

We next show a matching upper bound.
\ifarxiv

\begin{proposition}
    For any~$\kappa> \frac{n}{2n-1}$ and any~$n$, there exists a doubly stochastic~$n \times n$ matrix~$\mathcal{M}_{\kappa}$ such that~$(\mathcal{M},\kappa)$ is a no-instance of \dt.
\end{proposition}

\begin{proof}
    Let~$\varepsilon = \kappa - \frac{n}{2n-1} > 0$.
    We assume without loss of generality that~$n \geq 2$ as clearly no throughput of~$\frac{1}{2-1}+\varepsilon > 1$ is possible.
    Let~$\delta = \min(\frac{\varepsilon}{2},\frac{1}{2n-1})$.
    Then, we construct the demand matrix
    \[\mathcal{M}_{\kappa} = \begin{pmatrix}
        1-\delta & \frac{\delta}{n-1} & \dots & \frac{\delta}{n-1}\\
        \frac{\delta}{n-1} & 1-\delta & \dots & \frac{\delta}{n-1}\\
        \vdots & \vdots & \ddots & \frac{\delta}{n-1}\\
        \frac{\delta}{n-1} & \frac{\delta}{n-1}& \dots & 1-\delta
    \end{pmatrix}.\]
    Note that each row and column sums to~$1$.
    Consider any node~$v$.
    Since there is positive demand to all other nodes, any solution must contain at least one arc from~$v$ to each other node.
    Thus, at most~$2n-1-(n-1) = n$ arcs are of the form~$(v,v)$.
    We will show that this is insufficient to host~$(\frac{n}{2n-1}+\varepsilon)(1-\delta)$ from~$v$ to~$v$ directly.
    Note that~$n$ arcs can host at most~$\frac{n}{2n-1}$ flow directly.
    However,
    \begin{align*}
        (\frac{n}{2n-1}+\varepsilon)(1-\delta) &= \frac{n+(2n-1)\varepsilon-\delta n - (2n-1)\varepsilon \delta}{2n-1} \\
        &> \frac{n+(2n-1)\varepsilon-\varepsilon n - \varepsilon}{2n-1} \\
        &= \frac{n+ \varepsilon(2n-1-n-1)}{2n-1}\\
        &= \frac{n+ \varepsilon(n-2)}{2n-1}\\
        &\geq \frac{n}{2n-1}.
    \end{align*}
    The last inequality is due to the fact that~$n \geq 2$ and~$\varepsilon > 0$.
    This concludes the proof.
\end{proof}
\else
Due to space constraints, some of the remaining proofs are deferred to the full version~\cite{BAS26}.
Affected results are marked with a~$(\star)$.

\begin{proposition}[$\star$]
    For any~$\kappa> \frac{n}{2n-1}$ and any~$n$, there exists a doubly stochastic~$n \times n$ matrix~$\mathcal{M}_{\kappa}$ such that~$(\mathcal{M},\kappa)$ is a no-instance of \dt.
\end{proposition}
\fi

\section{Weak Throughput Bounds}
\label{sec:weak}
In this section, we show lower and upper bounds for the worst-case weak (direct) throughput of demand-aware networks.
We start with the lower bound.
We show the lower bound for \di{} but notice that the same lower bound also applies to \thr{} by \cref{obs:direct}.
Towards this goal, we first show an intermediate lemma and a simple observation.
Intuitively, the lemma shows that we can always round the entries in a matrix with entries between~$0$ and~$1$ such that the sums in any row and column remain unchanged and the average entry rounded up is at least as large as the average entry in the whole matrix.
\ifarxiv
\begin{lemma}
    \label{lem:average}
    Let~$\mathcal{M} = (a_{i,j})_{i,j \in [n]}$ such that~$0 \leq a_{i,j} \leq 1$ for all~${i,j \in [n]}$, $\sum_{j\in [n]} a_{i,j} = s_i \in \NN$ for all~$i \in [n]$, and~$\sum_{i\in [n]} a_{i,j} = t_j \in \NN$ for all~$j \in [n]$.
    Let~${X = \sum_{i \in [n]}s_i}$.
    Then, there exists a matrix~$\mathcal{N} = (b_{i,j})_{i,j \in [n]}$ such that~${b_{i,j} \in \{0,1\}}$ for all~$i,j \in [n]$, $\sum_{j\in [n]} b_{i,j} = s_i$ for all~$i \in [n]$, $\sum_{i\in [n]} b_{i,j} = t_j$ for all~$j \in [n]$, and~${\sum_{i \in [n]}\sum_{j\in [n]}(a_{i,j}\cdot b_{i,j}) \geq \frac{X^2}{n^2}}$. 
\end{lemma}

\begin{proof}
    We will build a sequence~$(\mathcal{C}_0,\mathcal{C}_1,\ldots,\mathcal{C}_\ell)$ of~$n \times n$ matrices where~$\mathcal{C}_0 = \mathcal{M}$, $C_\ell = \mathcal{N}$, and~$\mathcal{C}_h = (c^h_{i,j})_{i,j \in [n]}$ for each~$h \in [\ell]$ such that~$\sum_{j\in [n]} c^h_{i,j} = s_i$ for each~$h \in [\ell]$ and each~$i \in [n]$, $\sum_{j\in [n]} c^h_{i,j} = t_j$ for each~$h \in [\ell]$ and each~$j \in [n]$, and~$\sum_{i \in [n]}\sum_{j\in [n]}a_{i,j}c^h_{i,j} \geq \frac{X^2}{n^2}$ for each~$h \in [\ell]$.
    We will ensure that~$\mathcal{C}_{h+1}$ has at least one more integer entry (0 or 1) than~$\mathcal{C}_h$.
    Thus, $\ell \leq n^2$ and the procedure always terminates.
    Note that this will conclude the proof as~$\mathcal{C}_\ell = \mathcal{N}$ then satisfies all requirements of the lemma.
    
    We first show that~$\mathcal{C}_0 = \mathcal{M}$ satisfies the stated requirement.
    The sums of entries in row~$i$ and column~$j$ are by definition~$s_i$ and~$t_j$, respectively.
    Hence, it remains to show that~${\sum_{i \in [n]}\sum_{j\in [n]}a_{i,j}^2 \geq \frac{X^2}{n^2}}$.
    To this end, we first flatten the matrix into a vector, that is, we define~$z_p= a_{i,j}$ where~$p = in+j$.
    Note that~$\sum_{i=1}^{n^2}z_p = X$.
    Then, we use the Cauchy-Schwarz inequality~\cite{CS}, which states that for any two vectors~$\mathbf{x} = (x_1,x_2,\ldots,x_q)$ and~$\mathbf{y} = (y_1,y_2,\ldots,y_q)$ of non-negative numbers it holds that
    \[\Big(\sum_{i=1}^{q} (x_i \cdot y_i)\Big)^2 \leq \Big(\sum_{i=1}^q x_i^2\Big)\Big(\sum_{i=1}^q y_i^2\Big).\]
    Substituting~$x_i = z_i$ and~$y_i = 1$ for each~$i \in [n^2]$ and~$q=n^2$ yields
    \begin{align*}
        &&\Big(\sum_{i=1}^{n^2} z_i\Big)^2 &\leq \Big(\sum_{i=1}^{n^2} z_i^2\Big)\Big(\sum_{i=1}^{n^2} 1\Big)\\
        &\Leftrightarrow &X^2 &\leq \Big(\sum_{i \in [n]} \sum_{j \in [n]} a_{i,j}^2\Big) n^2\\
        &\Leftrightarrow &\sum_{i \in [n]} \sum_{j \in [n]} a_{i,j}^2 &\geq \frac{X^2}{n^2}.
    \end{align*}
    Note that this is precisely what we wanted to show.
    
    Finally, we show how to construct~$\mathcal{C}_{h+1}$ for a given matrix~$\mathcal{C}_h$.
    To this end, we build a bipartite graph~$G$ with a vertex~$u_r$ for each row~$r$ on one side and a vertex~$v_c$ for each column~$c$ on the other side.
    We add an edge~$\{u_r,v_c\}$ if and only if~$c^h_{r,c}$ is \emph{not} an integer.
    Note that no vertex is incident to exactly one edge as the sum of all entries in a row or column is an integer.
    Moreover, we assume that the graph contains at least one edge as otherwise we found~$\mathcal{C}_\ell = \mathcal{N}$ already.
    Then, the graph contains a (simple) cycle as we can start from an arbitrary edge and in each step pick an edge to extend the current path until a vertex is entered for the second time.
    The subpath between the two occurrences are then a simple cycle by definition.
    Let~$(e_1,e_2,\ldots,e_p)$ be the edges in the cycle and note that~$p$ is even as the graph is bipartite.
    Let~$E_o = \{e_1,e_3,\ldots,e_{p-1}\}$ and~${E_e = \{e_2,e_4,\ldots,e_p\}}$ be the set of the odd-indexed and even-indexed edges in the cycle, respectively. 
    Let~$\delta_o = \sum_{\{u_r,v_c\} \in E_o} a_{r,c}$ and~$\delta_e = \sum_{\{u_r,v_c\} \in E_e} a_{i,j}$.
    We then make a case distinction whether~$\delta_o \leq \delta_e$ or not.
    Since both cases are symmetric, we assume without loss of generality that~$\delta_o \leq \delta_e$.
    Let~$\varepsilon = \min(\min_{\{u_r,v_c\} \in E_o} c^h_{r,c},\min_{\{u_r,v_c\} \in E_e}(1-c^h_{r,c}))$.
    Note that~$\varepsilon > 0$.
    For each edge~$\{u_r,v_c\} \in E_o$, we set~$c^{h+1}_{r,c} = c^h_{r,c}-\varepsilon$, for each edge~$\{u_r,v_c\} \in E_e$, we set~$c^{h+1}_{r,c} = c^h_{r,c}+\varepsilon$, and for all other pairs~$(r,c)$, we set~$c^{h+1}_{r,c} = c^h_{r,c}$.
    Note that by definition of~$\varepsilon$, each entry remains in the interval~$[0,1]$ and at least one non-integer entry is replaced by either~$0$ or~$1$.
    Moreover, integer entries in~$\mathcal{C}_h$ remain unchanged in~$\mathcal{C}_{h+1}$.

    It remains to show that the sums in each row and column remains unchanged between~$\mathcal{C}_h$ and~$\mathcal{C}_{h+1}$ and~$\sum_{i \in [n]} \sum_{j\in [n]} c^h_{i,j} a_{i,j} \leq \sum_{i \in [n]} \sum_{j\in [n]} c^{h+1}_{i,j} a_{i,j}$.
    Towards the first statement, consider an arbitrary row~$r$.
    If~$u_r$ is not contained in the cycle, then no entries in row~$r$ are changed and so~$\sum_{j \in [n]} c^{h+1}_{r,j} = \sum_{j \in [n]}c^h_{r,j} = s_r$.
    If~$u_r$ is contained in the cycle, then the cycle contains exactly two edges incident to~$u_r$ and these two edges appear consecutively or are~$e_1$ and~$e_p$.
    In either case, exactly one of the edges is contained in~$E_o$ and the other is contained in~$E_e$.
    Hence, exactly one entry is increased by~$\varepsilon$ and one value is decreased by~$\varepsilon$.
    Thus, the sum within row~$r$ remains unchanged and the same argument also holds for each column.
    Finally, note that
    \begin{align*}
        \sum_{i \in [n]} \sum_{j\in [n]} (c^{h+1}_{i,j} a_{i,j}) &= \Big(\sum_{\{r,c\} \in E_o} (c^{h+1}_{r,c} a_{r,c})\Big) + \Big(\sum_{\{r,c\} \in E_e} (c^{h+1}_{r,c} a_{r,c})\Big) + \Big(\sum_{\{r,c\} \notin E_o \cup E_e} (c^{h+1}_{r,c} a_{r,c})\Big)\\
        &= \Big(\sum_{\{r,c\} \in E_o} ((c^{h}_{r,c}-\varepsilon) a_{r,c})\Big) + \Big(\sum_{\{r,c\} \in E_e} (c^{h}_{r,c}+\varepsilon) a_{r,c})\Big) + \sum_{\{r,c\} \notin E_o \cup E_e} (c^{h}_{r,c} a_{r,c})\\
        &= \Big(\sum_{i \in [n]} \sum_{j \in [n]} (c^{h}_{i,j} a_{i,j})\Big) + \varepsilon \Big(\Big(\sum_{\{r,c\} \in E_e} a_{r,c}\Big) - \Big(\sum_{\{r,c\} \in E_o}a_{r,c}\Big)\Big) \\
        &= \Big(\sum_{i \in [n]} \sum_{j \in [n]} (c^{h}_{i,j} a_{i,j})\Big) + \varepsilon (\delta_e - \delta_o) \\
        & \geq \sum_{i \in [n]} \sum_{j \in [n]} c^{h}_{i,j} a_{i,j}.
    \end{align*}
    This concludes the proof.
\end{proof}
\else
\begin{lemma}[$\star$]
    \label{lem:average}
    Let~$\mathcal{M} = (a_{i,j})_{i,j \in [n]}$ such that
    \begin{itemize}
        \item $0 \leq a_{i,j} \leq 1$ for all~${i,j \in [n]}$,
        \item $\sum_{j\in [n]} a_{i,j} = s_i \in \NN$ for all~$i \in [n]$, and
        \item $\sum_{i\in [n]} a_{i,j} = t_j \in \NN$ for all~$j \in [n]$.
    \end{itemize}
    Let~${X = \sum_{i \in [n]}s_i}$.
    Then, there exists a matrix~$\mathcal{N} = (b_{i,j})_{i,j \in [n]}$ such that~${b_{i,j} \in \{0,1\}}$ for all~$i,j \in [n]$, $\sum_{j\in [n]} b_{i,j} = s_i$ for all~$i \in [n]$, $\sum_{i\in [n]} b_{i,j} = t_j$ for all~$j \in [n]$, and~${\sum_{i \in [n]}\sum_{j\in [n]}(a_{i,j}\cdot b_{i,j}) \geq \frac{X^2}{n^2}}$.
\end{lemma}
\fi

We next show a simple observation that allows us to view the weak direct throughput from a slightly different angle, that is, we show an equivalent characterization.

\ifarxiv
\begin{observation}
    \label{lem:overlap}
    Given a doubly stochastic matrix~$\mathcal{M} = (a_{i,j})_{i,j \in [n]}$ and a value~$\kappa$, $(\mathcal{M},\kappa)$ is a yes-instance of \di{} if and only if there exists an $n \times n$ matrix~$\mathcal{N}=(b_{i,j})_{i,j \in [n]}$ where each entry is a non-negative integer and each row and each column sums to~$2n-1$ such that~$\sum_{i,j \in [n]} \min(\Np a_{i,j},b_{i,j}) \geq (2n-1)n\kappa$.
\end{observation}

\begin{proof}
    First, assume that~$(\mathcal{M},\kappa)$ is a yes-instance of \di.
    Then, there exists a matrix~$\mathcal{M}' = (a'_{i,j})_{i,j}$ where~$a'_{i,j} \leq a_{i,j}$ for all~$i,j \in [n]$ and a directed~$\Np$-regular graph~$G$ such that~$G$ can directly host~$\mathcal{M}'$ and~$\frac{\sum_{i,j \in [n]}a'_{i,j}}{n} \geq \kappa$.
    Let~$\mathcal{A} = (c_{i,j})_{i,j \in [n]}$ be the adjacency matrix of~$G$.
    We show that~$\mathcal{A}$ fulfills all requirements of matrix~$\mathcal{N}$ of the observation.

    First, since~$\mathcal{A}$ is the adjacency matrix of a directed~$\Np$-regular graph, all its entries are non-negative integers and all rows and columns sum to~$2n-1$ each.
    Second, since~$G$ can host~$\mathcal{M}'$ directly, it holds that~$c_{i,j} \geq (2n-1)a'_{i,j}$ and~$(2n-1)a_{i,j} \geq (2n-1)a'_{i,j}$.
    Thus,
    \[\sum_{i,j \in [n]} \min((2n-1)a_{i,j},c_{i,j}) \geq \sum_{i,j \in [n]} (2n-1)a'_{i,j} \geq (2n-1)n\kappa.\]

    In the other direction, assume a matrix~$\mathcal{N}$ as described by the observation exists.
    Then, let~$G$ be a graph such that~$\mathcal{N}$ is the adjacency matrix of~$G$.
    Note that~$G$ is a directed~$\Np$-regular graph by definition.
    Moreover, $G$ can host the demand matrix~${\mathcal{M}' = (a'_{i,j})_{i,j \in [n]} = (\min(a_{i,j},\frac{b_{i,j}}{2n-1}))_{i,j \in [n]}}$ directly.
    Note that~$a'_{i,j} \leq a_{i,j}$.
    
    Finally, $\sum_{i,j \in [n]} (2n-1) a'_{i,j} = \sum_{i,j \in [n]}\min(\Np a_{i,j},b_{i,j}) \geq (2n-1)n\kappa$ implies that~${\frac{\sum_{i,j \in [n]} a'_{i,j}}{n} \geq \frac{(2n-1)n\kappa}{(2n-1)n} = \kappa}$.
    Thus, $(\mathcal{M},\kappa)$ is a yes-instance of \di, concluding the proof.
\end{proof}
\else
\begin{observation}[$\star$]
    \label{lem:overlap}
    Let~$\mathcal{M} = (a_{i,j})_{i,j \in [n]}$ be a doubly stochastic matrix and let~$\kappa \in [0,1]$.
    Then, $(\mathcal{M},\kappa)$ is a yes-instance of \di{} if and only if there exists an $n \times n$ matrix~$\mathcal{N}=(b_{i,j})_{i,j \in [n]}$ where each entry is a non-negative integer and each row and each column sums to~$2n-1$ such that~$\sum_{i,j \in [n]} \min(\Np a_{i,j},b_{i,j}) \geq (2n-1)n\kappa$.
\end{observation}
\fi

We next show our lower bound for \di.

\begin{theorem}
    \label{thm:lowerbound}
    Let~$(\mathcal{M},\kappa)$ be an instance of \di{} where~$\mathcal{M}$ is an~$n \times n$ matrix.
    If~$\kappa \leq \frac{7n-4}{8n-4}$, then~$(\mathcal{M},\kappa)$ is a yes-instance.
\end{theorem}

\begin{proof}
    Let~$\mathcal{M} = (a_{i,j})_{i,j \in [n]}$.
    By \cref{lem:overlap}, we need to find a matrix~$\mathcal{N} = (b_{i,j})_{i,j \in [n]}$ where each entry~$b_{i,j}$ is a non-negative integer and each row and column sum to~$\N$ such that \[\sum_{i \in [n]}\sum_{j \in [n]}\min(\Np \cdot a_{i,j},b_{i,j}) \geq \Np n \frac{7n-4}{8n-4}.\]

    Let~$a'_{i,j}=(2n-1)a_{i,j}$ for all~$i,j \in [n]$ and~$\mathcal{M}'=(a'_{i,j})_{i,j \in [n]}$.
    We will construct~$\mathcal{N}$ as the sum of two matrices~$\mathcal{N}_1 + \mathcal{N}_2$.
    First, we set
    \[N_1 = (\lfloor \Np \cdot a_{i,j} \rfloor) = (p_{i,j})_{i,j \in [n]}\] as the integer part of~$\mathcal{M}'$.
    Let~$\mathcal{M}'' = \mathcal{M}' - \mathcal{N}_1 = (x_{i,j})_{i,j \in [n]}$ be the remaining part of~$\mathcal{M}'$.
    Note that~$0 \leq x_{i,j} < 1$ for each~${i,j \in [n]}$.
    Moreover, each row and column of~$\mathcal{M}''$ sums to an integer, which is at most~$n-1$, that is, the sum of all entries in~$\mathcal{M}''$ is at most~${\ifarxiv(n-1)n =\fi n^2-n}$.
    Let~$X$ be this value.
    We can now apply \cref{lem:average} to construct the matrix~${\mathcal{N}_2 = (q_{i,j})_{i,j \in [n]}}$ such that~${\sum_{i \in [n]} \sum_{j \in [n]} (x_{i,j} \cdot q_{i,j}) \geq \frac{X^2}{n^2}}$ and~$q_{i,j} \in \{0,1\}$ for each~$i,j \in [n]$.
    As said before, $\mathcal{N} = \mathcal{N}_1 + \mathcal{N}_2$.

    Note that each entry in~$\mathcal{N}$ is an integer by construction and the sum of entries in each row and column is exactly~$\Np$ as the sum in any given row or column is the same in~$\mathcal{M}''$ and in~$\mathcal{N}_2$.
    Since~$\mathcal{M}' = \mathcal{N}_1 + \mathcal{M}''$ and~$\mathcal{N} = \mathcal{N}_1 + \mathcal{N}_2$ and each row and column in~$\mathcal{M}'$ sums to exactly~$\Np$, the same also holds for~$\mathcal{N}$.
    Hence, it only remains to show that
    \[\sum_{i \in [n]}\sum_{j \in [n]}\min(\Np \cdot a_{i,j},b_{i,j}) \geq \Np n \frac{7n-4}{8n-4}.\]

    Note that whenever~$q_{i,j} = 0$, then
    \[\min(\Np \cdot a_{i,j},b_{i,j}) = p_{i,j} = p_{i,j} + q_{i,j} \cdot x_{i,j}\] and when~$q_{i,j}=1$, then
    \[\min(\Np \cdot a_{i,j},b_{i,j}) = \Np \cdot a_{i,j} = p_{i,j} + x_{i,j} = p_{i,j} + q_{i,j} \cdot x_{i,j}.\]
    This implies
    \begin{align*}
        &\quad \sum_{i \in [n]}\sum_{j \in [n]}\min(\Np \cdot a_{i,j},b_{i,j})\\
        &= \sum_{i \in [n]}\sum_{j \in [n]} (p_{i,j} + q_{i,j} \cdot x_{i,j})\\
        &= \Np n - X + \sum_{i \in [n]}\sum_{j \in [n]} (q_{i,j} \cdot x_{i,j})\\
        &\geq \Np n - X + \frac{X^2}{n^2}.
    \end{align*}
    Here, the last inequality is due to~\cref{lem:average}. 
    We now compute the minimum value for any~${X \leq n^2-n}$.
    \ifarxiv
    To this end, let~${f(x) = \Np n - x + \frac{x^2}{n^2}}$.
    \else
    Let~${f(x) = \Np n - x + \frac{x^2}{n^2}}$.
    \fi
    We compute the derivative~$f'(x) = \frac{2x}{n^2}-1$.
    Since~$f$ is clearly continuous, the minimum is attained at one of the two boundaries ($x=1$ or~$x=n^2-n$) or at a place where~$f'(x) = 0$.
    Note that the latter happens (only) at~$x=\frac{n^2}{2}$.
    For~$x=1$, we get
    \[f(1) = \Np n - 1 + \frac{1}{n^2} > \Np n - 1 \geq \Np n \frac{7n-4}{8n-4}\] for all~$n \geq 2$.
    The last inequality follows from the fact that~$\frac{7n-4}{8n-4} \geq \frac{5}{6}$ for all~$n \geq 2$ and~$\frac{\Np n - 1}{\Np n} \geq \frac{5}{6}$ for all~$n \geq 2$.
    Note that for~$n=1$, the only entry is integer and hence even~$\kappa=1$ yields a yes-instance.
    For~$x=n^2-n = (n-1)n$ (the upper bound for~$X$), we get
    \begin{align*}
    f((n-1)n) &= \Np n - (n-1)n + \frac{((n-1)n)^2}{n^2}\\
    &= \Np n - n^2 + n + (n-1)^2\\
    &= \Np n - n^2 + n + n^2 - 2n + 1 \geq \Np n -1
    \end{align*}
    for all~$n \geq 2$.
    Hence, the argument is the same as in the case~$x=1$.
    
    Finally, for~$x = \frac{n^2}{2}$, we get
    \begin{align*}
        f(\frac{n^2}{2}) &= \Np n - \frac{n^2}{2} + \frac{n^4}{4n^2}\\
        &= \Np n \big(1 - \frac{n^2}{4 \Np n}\big)\\
        &= \Np n \big(\frac{4 \Np n - n^2}{4 \Np n}\big)\\
        &= \Np n \big(\frac{8n^2-4n-n^2}{8n^2-4n}\big)\\
        &= \Np n \big(\frac{7n^2-4n}{8n^2-4n}\big)\\
        &= \Np n \frac{7n-4}{8n-4}.
    \end{align*}
    Hence, $X=\frac{n^2}{2}$ yields the minimum of~$\Np n \frac{7n-4}{8n-4}$.
    Thus, 
    \[\sum_{i \in [n]}\sum_{j \in [n]}\min(\Np \cdot a_{i,j},b_{i,j}) \geq \Np n \frac{7n-4}{8n-4},\] concluding the proof.
\end{proof}

To conclude this section, we show upper bounds for \thr{} and \di.

\ifarxiv
\begin{proposition}
    There exists a doubly stochastic~$2 \times 2$ matrix~$\mathcal{N}$ such that~$(\mathcal{N},\kappa)$ is a no-instance of \thr{} for any~$\kappa > \frac{8}{9}$.
    For each positive integer~$n$, there exists a doubly stochastic~$n \times n$ matrices~$\mathcal{M}_n$ such that~$(\mathcal{M}_n,\kappa)$ is a no-instance of \di{} for any~$\kappa > \frac{7n-4}{8n-4}$ when~$n$ is even and for any~$\kappa > \frac{7n-3}{8n-4}$ when~$n$ is odd.
\end{proposition}

\begin{proof}
    We start with the upper bound for \thr.
    Consider the matrix~$\mathcal{N} = \frac{1}{9} \cdot \begin{pmatrix}
        5 & 4\\
        4 & 5
    \end{pmatrix}.$
    Note that each row and column sums to exactly~$1$.
    There are only 4 different directed 3-regular graphs with 2 vertices, which are depicted in \cref{fig:fourgraphs} and repeated in \cref{fig:four} for convenience.
    \cref{fig:fourgraphs}.
        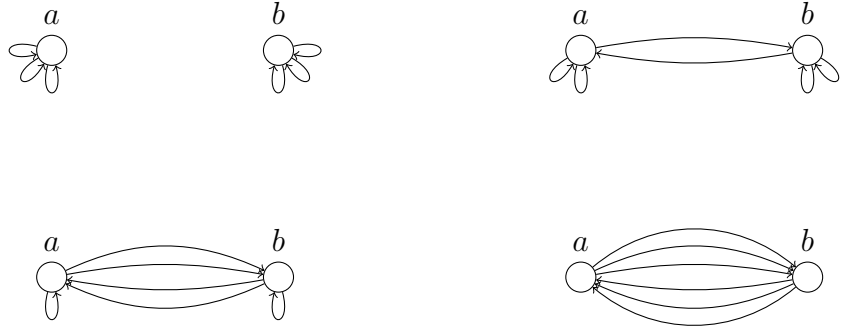
\begin{figure}[t]
        \centering
        \begin{tikzpicture}
            \def\x{3}
            \def\y{4}
            \foreach \i in {0,1}{
                \node[circle,draw,inner sep=4pt,label=$a$] at(\x*\i+\y*\i,0) (a\i) {};
                \node[circle,draw,inner sep=4pt,label=$b$] at(\x*\i+\y*\i+\x,0) (b\i) {};
            }
            \foreach \i in {2,3}{
                \node[circle,draw,inner sep=4pt,label=$a$] at(\x*\i+\y*\i-2*\x-2*\y,-3) (a\i) {};
                \node[circle,draw,inner sep=4pt,label=$b$] at(\x*\i+\y*\i-\x-2*\y,-3) (b\i) {};
            }
            \foreach \i in {0,1,2,3}{
                \foreach \j in {0,1,2,3}{
                    \ifthenelse{\j < \i}{
                        \draw[->,bend left=15*\j+10] (a\i) to (b\i);
                        \draw[->,bend left=15*\j+10] (b\i) to (a\i);
                    }{
                        \ifthenelse{\i = \j}{}{
                        \draw[->,in=45*\j+150,out=45*\j+120,loop] (a\i) to (a\i);
                        \draw[->,in=-45*\j+30,out=-45*\j+60,loop] (b\i) to (b\i);
                        }
                    }
                }
            }
        \end{tikzpicture}
        \caption{Repetition of \cref{fig:fourgraphs}. The four different directed $3$-regular graphs with two vertices are shown.}
        \label{fig:four}
    \end{figure}
    When each arc has capacity~$\frac{1}{3}$, the maximum fraction of~$\mathcal{N}$ that can be hosted by any graph is~$\frac{8}{9}$ which is achieved by both the graph in the top right and bottom left in \cref{fig:four} as shown next.
    The graph on the top right can host the matrix~$\frac{1}{9} \begin{pmatrix}
        5 & 3\\
        3 & 5
    \end{pmatrix}
    \leq \mathcal{N}$ using only direct arcs and the graph in the bottom left can host the matrix~$\frac{1}{9} \begin{pmatrix}
        5 & 4\\
        4 & 3
    \end{pmatrix} \leq \mathcal{N}$ where~$\frac{2}{9}$ flow from~$a$ to~$a$ is sent over the path~$((a,b),(b,a))$.

    We conclude with the upper bound for \di{}.
    Consider the matrix
    \[\mathcal{M}_n = \frac{1}{2n-1} \begin{pmatrix}
    2.5-a&1.5&2.5&1.5&\dots\\
    1.5&2.5-b&1.5&2.5&\dots\\
    2.5&1.5&2.5-a&1.5&\dots\\
    1.5&2.5&1.5&2.5-b&\dots\\
    \vdots&\vdots&\vdots&\vdots&\ddots
    \end{pmatrix},\]
    where $a=b=1$ whenever~$n$ is even and~$a=1.5$ and~$b=0.5$ whenever~$n$ is odd.
    We first show that each row and column adds to exactly~$1$.
    Whenever~$n$ is even, then each row sums to~$\frac{1}{2n-1} ((2.5+1.5)\frac{n}{2}-c)$, where~$c \in \{a,b\}$.
    That is, each row and column sums to~$\frac{1}{2n-1}(\frac{4n}{2}-1) = 1$.
    When~$n$ is odd, then each odd row or column sums to
    \[\frac{1}{2n-1}((2.5+1.5)\frac{n-1}{2}+2.5-a) = \frac{1}{2n-1}(2n - 2 + 2.5 - 1.5) = \frac{1}{2n-1}(2n-1) = 1.\]
    Each even row or column sums to
    \[\frac{1}{2n-1}((2.5+1.5)\frac{n-1}{2}+1.5-b) = \frac{1}{2n-1}(2n-2+1.5-0.5) = \frac{1}{2n-1}(2n-1) = 1.\]

    It remains to show that no directed~$(2n-1)$-regular graph can host a~$\kappa$-fraction of~$\mathcal{M}_n$ directly.
    Assume towards a contradiction that some graph~$G$ can host a~$\kappa$-fraction of~$\mathcal{M}_n$ directly.
    We distinguish between even and odd~$n$.
    We start with the even case and assume~$\kappa > \frac{7n-4}{8n-4}$.
    Note that for each row~$i$, at most~$\frac{3n}{2}-1$ arcs of~$G$ can be fully used and all remaining arcs can host at most~$\frac{1}{2(2n-1)}$~flow directly.
    Thus, it holds for the matrix~$\mathcal{M}_n'=(a'_{i,j})_{i,j \in [n]}$ that is hosted by~$G$ and where each entry is at most as large as in~$\mathcal{M}_n$ that row~$i$ sums to at most
    \[\frac{1}{2n-1}(\frac{3n}{2}-1+\frac{2n-1-(\frac{3n}{2}-1)}{2}) = \frac{3n-2+2n-1-\frac{3n}{2}+1}{2(2n-1)} = \frac{7n-4}{8n-4}.\]
    Since there are exactly~$n$ rows, it holds that~$\frac{\sum_{(i,j) \in V \times V} a'_{i,j}}{n} \leq \frac{n\cdot\frac{7n-4}{8n-4}}{n} = \frac{7n-4}{8n-4} < \kappa$.
    Thus, $G$ does not directly host a~$\kappa$-fraction of~$\mathcal{M}_n$, a contradiction.

    We conclude with the case where~$n$ is odd and assume~$\kappa> \frac{7n-3}{8n-4}$.
    At most~$3\cdot \frac{n-1}{2}+1$ arcs can be used fully and all remaining arcs can only be used with half capacity to host flow directly.
    This holds true since in odd rows, the number of fully usable arcs is bounded by~$3\cdot\frac{n-1}{2}+2.5-a = 3\cdot\frac{n-1}{2}+1$ and in even rows, the number is bounded by~$3\cdot\frac{n-1}{2}+1.5-b = 3\cdot\frac{n-1}{2}+1$.
    Then, it holds for the matrix~$\mathcal{M}_n'=(a'_{i,j})_{i,j \in [n]}$ that is hosted by~$G$ and where each entry is at most as large as in~$\mathcal{M}_n$ that any row sums to at most
    \begin{align*}
        &\frac{1}{2n-1}(3\cdot\frac{n-1}{2}+1+\frac{2n-1-(3\cdot\frac{n-1}{2}+1)}{2})\\
        =&\ \frac{3n-3+2+2n-1-\frac{3n-3}{2}-1}{2(2n-1)})\\
        =&\ \frac{6n-6+4+4n-2-3n+3-2}{4(2n-1)}\\
        =&\ \frac{7n-3}{8n-4}.
    \end{align*}
    Since there are~$n$ rows in~$\mathcal{M}_n'$, wee get~${\frac{\sum_{(i,j) \in V \times V} a'_{i,j}}{n} \leq \frac{n\cdot\frac{7n-3}{8n-4}}{n} = \frac{7n-3}{8n-4} < \kappa}$ by the same argument as above.
    Thus, $G$ does not directly host a~$\kappa$-fraction of~$\mathcal{M}_n$, a final contradiction.
    This concludes the proof.
\end{proof}
\else
\begin{proposition}[$\star$]
    There exists a doubly stochastic~$2 \times 2$ matrix~$\mathcal{N}$ such that~$(\mathcal{N},\kappa)$ is a no-instance of \thr{} for any~$\kappa > \frac{8}{9}$.
    For each positive integer~$n$, there exists a doubly stochastic~$n \times n$ matrices~$\mathcal{M}_n$ such that~$(\mathcal{M}_n,\kappa)$ is a no-instance of \di{} for any~$\kappa > \frac{7n-4}{8n-4}$ when~$n$ is even and for any~$\kappa > \frac{7n-3}{8n-4}$ when~$n$ is odd.
\end{proposition}
\fi

We mention in passing that our bounds for the weak direct throughput are tight for~${n=1}$ and all even~$n$, but there is a very small gap for all odd~$n \geq 3$.
We conjecture that the lower bound can be improved in these cases.

\section{Computational Complexity}
\label{sec:cc}
In this section, we analyze the computational complexity of the four problems we study.
We first show that \dt{} and \di{} are polynomial-time solvable.

\begin{proposition}
    \dt{} can be solved in~$O(n^3)$ time.
\end{proposition}

\begin{proof}
    We show that the greedy algorithm that was already used in the proof of \cref{prop:directstronglower} solves \dt{} optimally.
    We first recall the algorithm.
    To this end, let~$\mathcal{M} = (a_{i,j})_{i,j \in [n]}$ be the input matrix.
    The algorithm builds the solution graph~$G$ by describing its adjacency matrix~${\mathcal{A} = (b_{i,j})_{i,j \in [n]}}$.
    We start with~$b_{i,j}=0$ for all~$i,j \in [n]$.
    We then compute~$\frac{a_{i,j}}{b_{i,j}}$ for all~$i,j \in [n]$ and pick an entry~$i^*,j^*$ where this value is maximum and increase~$b_{i^*,j^*}$ by one until the sum in a row or column reaches~$2n-1$ at which point all entries in this row or column are ignored.
    Consider an entry where~$\frac{a_{i,j}}{b_{i,j}}$ is maximum.
    If~$b_{i,j}$ is not increased, then in order to directly host~$\kappa \mathcal{M}$, the throughput is at most~$\kappa \leq \frac{a_{i,j}}{b_{i,j}}$.
    By the choice of the entry to increase, the throughput is at most~$\frac{a_{i,j}}{b_{i,j}}$ and thus $b_{i,j}$ has to be increased or all further arc additions do not improve the throughput.
    Since the order in which arcs are added does not matter, it is always optimally to increase~$b_{i,j}$ as the algorithm does.

    We next analyze the running time.
    Initially, we compute~$n^2$ entries~$\frac{a_{i,j}}{b_{i,j}}$.
    We can sort all of them in~$O(n^2 \log(n^2)) = O(n^2 \log(n))$ time.
    We also keep track of the total score in each row and column.
    We then increase a value~$b_{i,j}$ for a total of~$n (2n-1)$ times.
    Each time, we compute one new entry and insert it into the ordered list of entries, increase the score for one row and one column and if the score reaches~$2n-1$, then we remove all entries from that row or column from the sorted list.
    The time to insert one entry is~$O(\log(n^2))$ using binary search and the time to remove all entries for one row or column is~$O(n^2)$.
    Since we do the former~$(2n-1)n$ times and the latter~$2n$ times in total, the overall running time is in~$O(n^3)$.
\end{proof}

We continue with \di.

\begin{proposition}
    \di{} is solvable in ${O(n^4 \log(n))}$ time.
\end{proposition}

\begin{proof}
Let~$(\mathcal{M} = (a_{i,j})_{i,j \in [n]},\kappa)$ be an instance of \di.
By \cref{lem:overlap}, $(\mathcal{M},\kappa)$ is a yes-instance if and only if there exists an~$n \times n$ matrix~${\mathcal{N}= (b_{i,j})_{i,j \in [n]}}$ where each entry~$b_{i,j}$ is a non-negative integer, each row and each column sums to~$2n-1$, and \[\sum_{i,j \in [n]} \min((2n-1) a_{i,j},b_{i,j}) \geq (2n-1) n \kappa.\]
Let~$a'_{i,j} = (2n-1) a_{i,j}$ for all~$i,j \in [n]$.
We show how to find a matrix~$\mathcal{N}$ as described above in polynomial time using maximum-cost flow.

We construct a directed graph with vertex set
\[
V = \{s,t\} \cup \{r_i \mid i \in [n]\} \cup \{c_j \mid j \in [n]\}.
\]
For each $i \in [n]$, we add an arc~$(s,r_i)$ with capacity $2n-1$ and cost~0, and for each $j \in [n]$, add an arc~$(c_j,t)$ with capacity $2n-1$ and cost~0.

For each~$i,j \in [n]$, we add three parallel arcs from $r_i$ to $c_j$ as follows.
\begin{itemize}
    \item An arc with capacity~$\lfloor a'_{i,j} \rfloor$ and cost~1,
    \item an arc with capacity~1 and cost $a'_{i,j}-\lfloor a'_{i,j} \rfloor$, and
    \item an arc with capacity $2n-1$ and cost~0.
\end{itemize}
\ifarxiv
An example of this construction can be seen in \cref{fig:flowexample}.
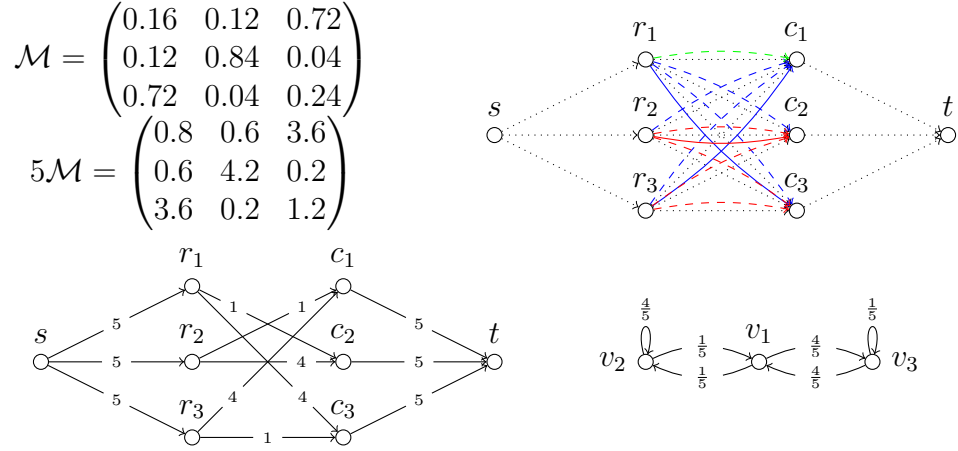
\begin{figure*}[t]
    \centering
    \begin{tikzpicture}
        \node at(-2,5) {$\mathcal{M} = \begin{pmatrix}
        0.16 & 0.12 & 0.72\\
        0.12 & 0.84 & 0.04\\
        0.72 & 0.04 & 0.24
    \end{pmatrix}$};
    
    \node at(-2,3.5) {$5\mathcal{M} =\begin{pmatrix}
        0.8 & 0.6 & 3.6\\
        0.6 & 4.2 & 0.2\\
        3.6 & 0.2 & 1.2
    \end{pmatrix}$};
    
    \node[circle,draw,inner sep=2pt, label=$s$] at (2,4) (s) {};
    \node[circle,draw,inner sep=2pt, label=$r_1$] at (4,5) (r1) {};
    \node[circle,draw,inner sep=2pt, label=$r_2$] at (4,4) (r2) {};
    \node[circle,draw,inner sep=2pt, label=$r_3$] at (4,3) (r3) {};
    \node[circle,draw,inner sep=2pt, label=$c_1$] at (6,5) (c1) {};
    \node[circle,draw,inner sep=2pt, label=$c_2$] at (6,4) (c2) {};
    \node[circle,draw,inner sep=2pt, label=$c_3$] at (6,3) (c3) {};
    \node[circle,draw,inner sep=2pt, label=$t$] at (8,4) (t) {};
    \foreach \i in {1,2,3}{
        \draw[dotted,->] (s) to (r\i);
        \draw[dotted,->] (c\i) to (t);
        \foreach \j in {1,2,3}{
            \draw[dotted,->] (r\i) to (c\j);
        }
    }
    \draw[->,blue,bend right=10] (r1) to (c3);
    \draw[->,blue,bend right=10] (r3) to (c1);
    \draw[->,red,bend right=10] (r2) to (c2);
    \draw[dashed,bend left=10,->,blue] (r1) to (c2);
    \draw[dashed,bend left=10,->,blue] (r1) to (c3);
    \draw[dashed,bend left=10,->,blue] (r2) to (c1);
    \draw[dashed,bend left=10,->,blue] (r3) to (c1);
    \draw[dashed,bend left=10,->,red] (r2) to (c2);
    \draw[dashed,bend left=10,->,red] (r2) to (c3);
    \draw[dashed,bend left=10,->,red] (r3) to (c2);
    \draw[dashed,bend left=10,->,red] (r3) to (c3);
    \draw[dashed,bend left=10,->,green] (r1) to (c1);

    \node[circle,draw,inner sep=2pt, label=$s$] at (-4,1) (s) {};
    \node[circle,draw,inner sep=2pt, label=$r_1$] at (-2,2) (r1) {};
    \node[circle,draw,inner sep=2pt, label=$r_2$] at (-2,1) (r2) {};
    \node[circle,draw,inner sep=2pt, label=$r_3$] at (-2,0) (r3) {};
    \node[circle,draw,inner sep=2pt, label=$c_1$] at (0,2) (c1) {};
    \node[circle,draw,inner sep=2pt, label=$c_2$] at (0,1) (c2) {};
    \node[circle,draw,inner sep=2pt, label=$c_3$] at (0,0) (c3) {};
    \node[circle,draw,inner sep=2pt, label=$t$] at (2,1) (t) {};
    
    \draw[->] (s) to node[circle,fill=white,midway,inner sep=2pt] {\tiny $5$} (r1);
    \draw[->] (s) to node[circle,inner sep=2pt,fill=white,midway] {\tiny $5$} (r2);
    \draw[->] (s) to node[circle,inner sep=2pt,fill=white,midway] {\tiny $5$} (r3);
    \draw[->] (c1) to node[circle,inner sep=2pt,fill=white,midway] {\tiny $5$} (t);
    \draw[->] (c2) to node[circle,inner sep=2pt,fill=white,midway] {\tiny $5$} (t);
    \draw[->] (c3) to node[circle,inner sep=2pt,fill=white,midway] {\tiny $5$} (t);
    \draw[->] (r1) to node[circle,inner sep=2pt,fill=white,near start] {\tiny $1$} (c2);
    \draw[->] (r1) to node[circle,inner sep=2pt,fill=white,near end] {\tiny $4$} (c3);
    \draw[->] (r2) to node[circle,inner sep=2pt,fill=white,near end] {\tiny $1$} (c1);
    \draw[->] (r2) to node[circle,inner sep=2pt,fill=white,near end] {\tiny $4$} (c2);
    \draw[->] (r3) to node[circle,inner sep=2pt,fill=white,midway] {\tiny $1$} (c3);
    \draw[->] (r3) to node[circle,inner sep=2pt,fill=white,near start] {\tiny $4$} (c1);

    \node[circle,inner sep=2pt,draw,label=$v_1$] at(5.5,1) (v1) {};
    \node[circle,inner sep=2pt,draw,label=left:$v_2$] at(4,1) (v2) {};
    \node[circle,inner sep=2pt,draw,label=right:$v_3$] at(7,1) (v3) {};
    \draw[->,bend left=25] (v1) to node[circle,inner sep=2pt,fill=white,midway] {\tiny $\frac{1}{5}$} (v2);
    \draw[->,bend left=25] (v1) to node[circle,inner sep=2pt,fill=white,midway] {\tiny $\frac{4}{5}$} (v3);
    \draw[->,bend left=25] (v2) to node[circle,inner sep=2pt,fill=white,midway] {\tiny $\frac{1}{5}$} (v1);
    \draw[->,bend left=25] (v3) to node[circle,inner sep=2pt,fill=white,midway] {\tiny $\frac{4}{5}$} (v1);
    
    \draw[->,loop above] (v3) to (v3);
    \draw[->,loop above] (v2) to (v2);
    \node at(4,1.7) {\tiny $\frac{4}{5}$};
    \node at(7,1.7) {\tiny $\frac{1}{5}$};
    \end{tikzpicture}
    \caption{The top left shows an input demand matrix~$\mathcal{M}$.
    Below is the matrix~${(2n-1)\mathcal{M}}$. The top right shows the constructed instance of maximum-cost flow, where dotted arcs have capacity~$2n-1=5$ and cost~$0$, dashed arcs have capacity one and costs depending on their color: green for~$0.8$, blue for costs~$0.6$, and red for~$0.2$.
    The fully drawn arcs have cost~$1$ and capacity based on their color: blue for capacity~$3$ and red for capacity~$4$.
    The bottom left shows an optimal solution for the constructed instance corresponding to the solution graph depicted in the bottom left (with the shown capacities and each time the highest cost arc(s) are used). This graph can directly host the matrix 
    $\begin{pmatrix}
        0 & 0.12 & 0.72\\
        0.12 & 0.8 & 0\\
        0.72 & 0 & 0.2
    \end{pmatrix}$.
    }
    \label{fig:flowexample}
\end{figure*}
\fi
We compute a maximum-cost flow of value~$n(2n-1)$ and check whether the cost is at least~$(2n-1) n \kappa$.

We next show that the construction is correct.
To this end, first assume that there exists a matrix~$\mathcal{N}=(b_{i,j})_{i,j \in [n]}$ such that each~$b_{i,j}$ is a non-negative integer, each row and column in~$\mathcal{N}$ sums to~${2n-1}$, and~${\sum_{i,j \in [n]} \min(a'_{i,j},b_{i,j}) \geq (2n-1) n \kappa}$.
For each~${i \in [n]}$, we send~${2n-1}$ units of flow over the arcs~$(s,r_i)$ and~$(c_i,t)$.
These have cost~$0$.
Moreover for each~$i,j \in [n]$, we send~$b_{i,j}$ units of flow from~$r_i$ to~$c_j$.
If~$b_{i,j} \leq a'_{i,j}$, then these have cost~$b_{i,j}$.
If~$b_{i,j} > a'_{i,j}$, then these have total cost~$a'_{i,j}$ by construction.
By definition of~$\mathcal{N}$, there are exactly~$2n-1$~units of flow leaving each~$r_i$ and~${2n-1}$~units of flow entering~$c_j$.
Thus, there is an~$n(2n-1)$-flow from~$s$ to~$t$ of cost~$\sum_{i,j \in [n]} \min(a'_{i,j},b_{i,j})$.

In the other direction, assume that an~$n(2n-1)$-flow from~$s$ to~$t$ of cost~${x \geq (2n-1)n\kappa}$ exists.
Since all capacities and~$n(2n-1)$ are integers, we may assume that the flow over any arc is integral.
Now, we build~$\mathcal{N}=(b_{i,j})_{i,j \in [n]}$ where each~$b_{i,j}$ is an integer, each row and each column sums to~$2n-1$, and~${\sum_{i,j \in [n]} \min((2n-1) a_{i,j},b_{i,j}) = x}$
(and recall that~${x \geq (2n-1)n\kappa}$).
Let~$b_{i,j}$ be the amount of flow send over arcs of the form~$(r_i,c_j)$.
By construction, the total amount of flow leaving any vertex~$r_i$ and the total amount of flow entering~$c_j$ is~$2n-1$.
Hence, each row and each column in~$\mathcal{N}$ sums to~$2n-1$.
Finally for a fixed pair $(i,j)$, the cost incurred by sending~$b_{i,j}$~units of flow from $r_i$ to $c_j$ is exactly~$\min(a'_{i,j}, b_{i,j})$~as shown next.
The first~$\lfloor a'_{i,j} \rfloor$~units contribute cost~1 each, the next unit contributes~${a'_{i,j}-\lfloor a'_{i,j} \rfloor}$, and all remaining units contribute cost~0.
Thus, the total cost of the flow equals $\sum_{i,j \in [n]} \min(a'_{i,j}, b_{i,j})$.
Since the cost of the flow is~$x \geq (2n-1)n\kappa$ by definition, we have \[\sum_{i,j \in [n]} \min(a'_{i,j}, b_{i,j}) = x \geq (2n-1)n\kappa,\] concluding the proof of correctness.

Finally, we analyze the running time.
The constructed graph has~${N=2n+2}$~vertices and~${M \in \Theta(n^2)}$ arcs.
It can be built in time~${O(N+M) = O(n^2)}$.
A maximum-cost~${n(2n-1)}$-flow can be computed in~$O(n(2n-1)M\log(N)) = O(n^4 \log(n))$ time~\cite{DBLP:books/networkflows}, completing the proof.
\end{proof}

To conclude this section, we next show that \thr{} is NP-hard.
\ifarxiv
\begin{theorem}
    \label{thm:NPhard}
    \thr{} is NP-hard.
\end{theorem}

\begin{proof}
    We present a reduction from the NP-hard problem \textsc{Exact Cover by 3-Sets}~\cite{GJ79}.
    Therein, one is given a universe~$U = \{x_1,x_2,\ldots,x_N\}$ and a family~${\mathcal{F} = \{S_1,S_2,\ldots,S_M\}}$ of subsets of~$U$.
    Each set~$S_j \subseteq U$ has size exactly three and the question is whether there is a subfamily~$\mathcal{F}' \subseteq \mathcal{F}$ such that each element~$x_i \in U$ is contained in exactly one set~$S_j \in \mathcal{F}'$.
    Note that we can assume without loss of generality that~$N = 3K$ for an integer~$K$ and the size of~$\mathcal{F}'$ is precisely~$K$ whenever a solution exists.
    Thus, we can also assume that~$M \geq K$.
    For each element~$x_j$, we define~$\alpha_j$ to be the number of sets~$S_i$ that contain~$x_j$.
    
    We construct the demand matrix of a network where the number of vertices is
    \[n= 2+2N + M + 10 (M-K) + 10 \Big(\frac{M}{2} + \frac{N}{6} - K\Big) + 15 M - 10 K + 5M.\]
    Let~$V= \{s,t\} \cup \{u_i,v_i\mid i \in [M]\} \cup \{y_i \mid i \in [N]\} \cup A \cup B \cup Z \cup \bigcup_{i \in [M]} W_i$ be the set of vertices, where
    \begin{align*}
        A&=\{a_i \mid i \in [10 (M-K)]\},\\
        B&=\{b_i \mid i \in [10 \Big(\frac{M}{2} + \frac{N}{6} - K\Big)]\},\\
        Z&= \{z_i \mid i \in [15M - 10K]\},\text{ and}\\
        W_i&=\{w_j^i \mid j \in [5]\}.
    \end{align*}
    We next state all demands, which are also depicted graphically in~\cref{fig:reduction}.
    \begin{figure}[t]
        \centering
        \begin{tikzpicture}[scale=0.97]
            \node[circle,draw,inner sep=3pt] (s) at (-4.5,0) {$s$};
            \node[circle,draw,inner sep=3pt] (t) at (4.5,0) {$t$};
            \node[circle,draw,inner sep=3pt] (ui) at (-1.5,0) {$u_i$};
            \node[circle,draw,inner sep=3pt] (vi) at (1.5,0) {$v_i$};
            \node[circle,draw,inner sep=3pt] (yj) at (3,1.75) {$y_j$};
            \node[ellipse,draw,minimum width=8cm,minimum height=1cm] (Z) at (0,4) {$Z$};
            \node[ellipse,draw,minimum width=2cm,minimum height=1cm] (W) at (-1,1.75) {$W_i$};
            \node[ellipse,draw,minimum width=1cm,minimum height=2cm] (A) at (-6.5,0) {$A$};
            \node[ellipse,draw,minimum width=1cm,minimum height=2cm] (B) at (6.5,0) {$B$};

            \draw[->,loop below] (s) to node[below]{$\ns-|A|-K$} (s);
            \draw[->,loop below] (t) to node[below]{$\ns-|B|-K$} (t);
            \draw[->,loop above] (A) to node[above]{$\ns-1$} (A);
            \draw[->,loop above] (B) to node[above]{$\ns-1$} (B);
            \draw[->,loop left] (W) to node[left]{$\ns-1$} (W);
            \draw[->,loop above] (Z) to node[above]{$\ns-\frac{1}{10}$} (Z);
            \draw[->,loop below] (ui) to node[xshift=-25pt,yshift=15pt]{$\ns-\frac{11}{2}$} (ui);
            \draw[->,loop below] (vi) to node[xshift=25pt,yshift=15pt]{$\ns-4$} (vi);
            \draw[->,loop right] (yj) to node[xshift=-25pt,yshift=15pt]{$\ns-\alpha_j$} (yj);
            \draw[->,bend right=15] (Z.south west) to node[above left]{$\frac{1}{10|Z|}$} (A);
            \draw[->,bend right=15] (B) to node[above right]{$\frac{1}{10|Z|}$} (Z.south east);
            \draw[->,bend left=40] (t) to node[above]{$K$} (s);
            \draw[->] (s) to node[above]{$\frac{1}{2}$} (ui);
            \draw[->,bend right=30] (s) to node[below]{$\frac{1}{2}$} (vi);
            \draw[->,bend right=5] (ui) to node[below]{$\frac{1}{2}$} (vi);
            \draw[->,bend right=5] (vi) to node[above]{$\frac{1}{2}$} (ui);
            \draw[->] (vi) to node[above]{$\frac{1}{2}$} (t);
            \draw[->,bend left=10] (ui) to node[above left]{$1$} (W);
            \draw[->,bend left=10] (W) to node[xshift=15pt,yshift=5pt]{$\frac{9}{10}$} (ui);
            \draw[->,bend left=5] (W) to node[left]{$\frac{1}{10|Z|}$} (Z);
            \draw[->,bend left=10] (Z) to node[right]{$\frac{1}{2|Z|}$} (vi);
            \draw[->,bend right=10] (s) to node[above]{$\frac{9}{10}$} (A);
            \draw[->,bend right=10] (A) to node[below]{$1$} (s);
            \draw[->,bend right=10] (t) to node[below]{$1$} (B);
            \draw[->,bend right=10] (B) to node[above]{$\frac{9}{10}$} (t);
            \draw[->] (yj) to node[right]{$\frac{1}{6}$} (t);
            \draw[->,bend right=10] (yj) to node[right]{$\frac{\alpha_j-1}{6|Z|}$} (Z);
            \draw[->,bend right=10,dashed] (vi) to node[xshift=10pt]{$1$} (yj);
            \draw[->,bend right=10,dashed] (yj) to node[above left]{$\frac{5}{6}$} (vi);
        \end{tikzpicture}
        \caption{A graphical representation of the demand matrix constructed in the proof of \cref{thm:NPhard} (where all demands are multiplied with~$\N$). The large nodes~$A,B,W_i$, and~$Z$ represent sets of vertices and an arc to or from such a set represents one such link to each vertex in the set (and an arc between two such sets represents all pairwise arcs, except for self-loops which are only present from each vertex in the set to itself). For the sake of minimizing visual clutter, only the vertices corresponding to a single element~$x_j$~($y_j$) and a single set~$S_i$ ($u_i$, $v_i$, and~$W_i$) are shown. The dashed arcs between~$y_j$ and~$v_i$ only exist if~$x_j \in S_i$ and there are no demands between vertices corresponding to two different elements or to two different sets.}
        \label{fig:reduction}
    \end{figure}
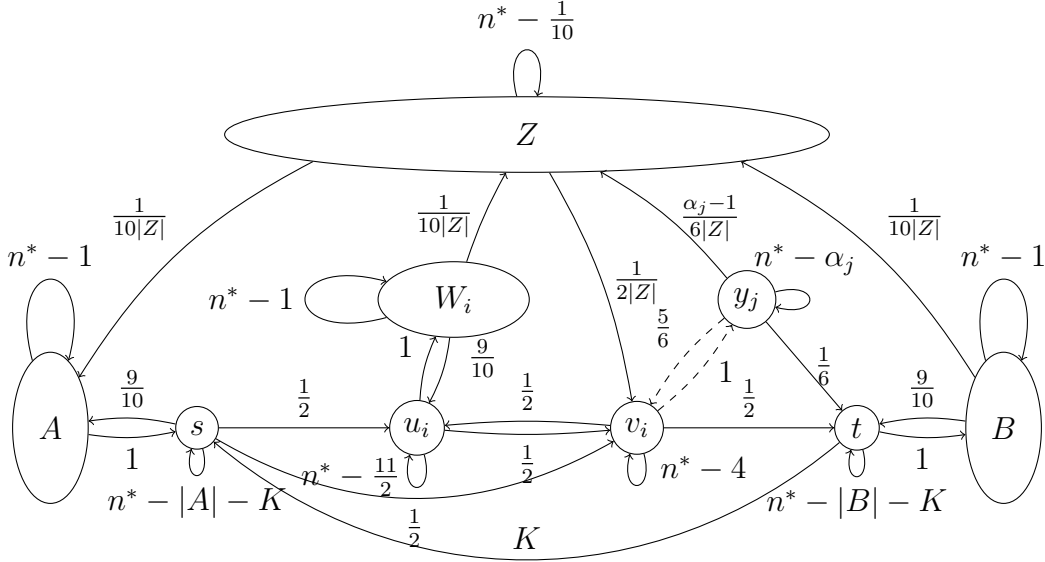%
    To this end, let~$\ns = 2n-1$ and we present all demands multiplied by~$\ns$ for notational convenience.
    The scaled demands~$d_{u,v}$ from a vertex~$u$ to a vertex~$v$ are defined as follows and for the sake of readability, we will only define the demands that are larger than~$0$ and implicitly assume that all demands not mentioned in the following are~$0$.
    We start with demand originating from vertex~$s$ and set~$d_{s,s} = \ns - |A| - K$, $d_{s,u_i} = d_{s,v_i} = \frac{1}{2}$ for all~$i \in [M]$, and~$d_{s,a_i}=\frac{9}{10}$ for all~$i \in [|A|]$.
    The demand originating from~$t$ is~$d_{t,t}=\ns - |B| - K$, $d_{t,s} = K$, and~$d_{t,b_i} = 1$ for all~$i \in [|B|]$.
    The demand originating from a vertex~$u_i$ is defined as~$d_{u_i,u_i} = \ns - \frac{11}{2}$, $d_{u_i,v_i} = \frac{1}{2}$, and~$d_{u_i,w^i_j} = 1$ for all~$j \in [n]$.
   The demand originating from a vertex~$v_i$ is~$d_{v_i,t} = d_{v_i,u_i} = \frac{1}{2}$, $d_{v_i,v_i} = \ns - 4$, and~$d_{v_i,y_j} = \frac{5}{6}$ for all~$i \in [M]$ and all~$j$ such that~$x_j \in S_i$.
    We next define the demand originating from each vertex~$y_j$ and set~$d_{y_j,t}=\frac{1}{6}$, $d_{y_j,v_i}=\frac{5}{6}$ whenever~$x_j \in S_i$, $d_{y_j,y_j} = \ns - \alpha_j$, and~$d_{y_j,z_i} = \frac{\alpha_j-1}{6|Z|}$ for all~$i \in [|Z|]$.
    For each vertex~$w^i_j$ with~$i \in [M]$ and~$j \in [5]$, we set~$d_{w^i_j,w^i_j} = \ns - 1$, $d_{w^i_j,u_i} = \frac{9}{10}$, and~$d_{w^i_j,z_\ell} = \frac{1}{10|Z|}$ for each~$\ell \in [|Z|]$.

    It remains to define the demands originating from vertices in~$A \cup B \cup Z$.
    For each~${i \in [|A|]}$, we set~$d_{a_i,a_i} = \ns - 1$ and~$d_{a_i,s}= 1$.
    For each~$i \in [B]$, we define~$d_{b_i,b_i} = \ns - 1$, $d_{b_i,t}=\frac{9}{10}$, and~$d_{b_i,z_j} = \frac{1}{10|Z|}$ for each~$j \in [|Z|]$.
    Finally, for each~$i \in [|Z|]$, we set~$d_{z_i,z_i} = \ns -\frac{1}{10}$, $d_{z_i,a_j} = \frac{1}{10|Z|}$ for each~$j \in |A|$, and~$d_{z_i,v_j} = \frac{1}{2|Z|}$ for each~$j \in [M]$.
    To conclude the construction, we define a couple of functions and values.
    For a number~$x$, let
    \[f(x) = \begin{cases} x & \text{ if } x - \lfloor x \rfloor \geq \frac{3}{4}\\ \lfloor x \rfloor & \text{ otherwise}\end{cases} \text{ and } g(x) = \begin{cases} \lceil x \rceil & \text{ if } x - \lfloor x \rfloor \geq \frac{3}{4}\\ \lfloor x \rfloor & \text{ otherwise.}\end{cases}\]
    Let~$H = \sum_{u \in V} \sum_{v \in V} f(d_{u,v})$ and~$L = n \ns - \sum_{u \in V}\sum_{v \in V}g(d_{u,v})$.
    We set~$\kappa = \frac{H + \frac{3L}{4} + \frac{N}{12}}{n \ns}$.
    This concludes the construction.

    We next show that the demand originating and terminating in any vertex is~$1$, that is, the scaled demand is~$\ns$.
    Afterwards, we show that there exists a directed~$\ns$-regular graph that can host a~$\kappa$-fraction of~$\frac{1}{2n-1}\mathcal{M}$ (where~$\mathcal{M} = (\frac{d_{i,j}}{\ns})_{i,j \in [n]}$) if and only if there exists an exact cover for the original instance.
    Since the construction can clearly be computed in polynomial time and no demand is negative, this concludes the proof.

    It is easy to verify that the scaled demand originating from any vertex other than~$s, y_j$, or~$z_i$ for~$i \in [|Z|]$ and~$j \in [N]$ is~$\ns$.
    We will skip the formal proof for the sake of conciseness and focus on the slightly non-trivial parts.
    In a similar manner, all scaled demands terminating in any vertex other than~$t$, $v_i$, and~$z_j$ for~$i \in [M]$ and~$j \in [|Z|]$ are trivially~$\ns$ and we focus on the listed vertices.
    The total scaled demand originating in~$s$ is
    \[(\ns - |A| - K) + M\cdot \frac{1}{2} + M\cdot \frac{1}{2} + |A| \cdot \frac{9}{10} = \ns - \frac{|A|}{10} + M - K = \ns - (M - K) + M - K = \ns.\]
    The total scaled demand originating in a vertex~$y_j$ is
    \[\frac{1}{6} + \alpha_j \cdot \frac{5}{6} + (\ns - \alpha_j) + |Z| \cdot \frac{\alpha_j - 1}{6|Z|} = \ns + \frac{1 + 5 \alpha_j - 6 \alpha_j + \alpha_j - 1}{6} = \ns.\]
    The total scaled demand originating from a vertex~$z_i$ is \[(\ns - \frac{1}{10}) + |A| \cdot \frac{1}{10|Z|} + M \cdot \frac{1}{2|Z|} = \ns - \frac{1}{10} + \frac{|A|+5M}{10|Z|} = \ns - \frac{1}{10} + \frac{10M - 10K + 5M}{10(15M - 10K)} = \ns.\]
    The total scaled demand terminating in~$t$ is
    \begin{align*}
        (\ns - |B| - K) + M \cdot \frac{1}{2} + |B| \cdot \frac{9}{10} + N \cdot \frac{1}{6} &= \ns - \frac{|B|}{10} - K + \frac{M}{2} + \frac{N}{6}\\
        &= \ns - (\frac{M}{2} + \frac{N}{6} - K) - K + \frac{M}{2} + \frac{N}{6} = \ns.
    \end{align*}
    We next analyze the demand terminating in a vertex~$v_i$.
    Note that for each set~$S_i$, there are exactly three elements~$x_j$ such that~$x_j \in S_i$.
    Hence, the total scaled demand terminating in a vertex~$v_i$ is
    \[\frac{1}{2} + \frac{1}{2} + (\ns - 4) + |Z| \cdot \frac{1}{2|Z|} + 3 \cdot \frac{5}{6} = \ns - 4 + 1 + \frac{1}{2} + \frac{15}{6} = \ns.\]
    Finally, the total scaled demand terminating in a vertex~$z_j$ is
    \begin{align*}
    &\ \ \ (\ns - \frac{1}{10}) + \Big(\sum_{i \in [M]} 5 \cdot \frac{1}{10|Z|}\Big) + \Big(\sum_{j \in [N]} \frac{\alpha_j-1}{6|Z|}\Big) + |B| \cdot \frac{1}{10|Z|}\\
    &= \ns - \frac{1}{10} + \frac{5M+\frac{10}{6}(3M-N)+|B|}{10|Z|}\\
    &= \ns - \frac{1}{10} + \frac{5M + 5M - \frac{10N}{6} + 10(\frac{M}{2}+\frac{N}{6} - K)}{10 (15M - 10K)}\\
    &= \ns - \frac{1}{10} + \frac{15M - \frac{10N}{6} + \frac{10N}{6} - 10K}{10 (15M - 10K)} = \ns - \frac{1}{10} + \frac{1}{10} = \ns.
    \end{align*}
    For the first equality, note that~$\sum_{j \in [N]} \alpha_j = 3M$ as each set contains exactly three elements.
    Thus, the instance we constructed is valid and it remains to prove that it is correct.

    Before we formally prove the correctness of our construction, we first give a high-level intuition.
    Given a solution, that is, a directed~$\ns$-regular graph and a flow~$\mathcal{P}$, we will define the \emph{contribution} of an arc~$e$ as~$\sum_{(P_i,d_i) \in \mathcal{P}} \frac{d_i}{\len(P)} \cdot \mathbf{1}_{e \in P_i}$ and show that a solution has to contain all arcs that can contribute more than~$\frac{3}{4}$ and all arcs have to contribute at least~$\frac{3}{4}$.
    This contains all arcs that can directly host a flow of at least~$\frac{3}{4}$ and by construction, there are no pairs~$(u,v)$ where the demand from~$u$ to~$v$ is in the interval~$(\frac{1}{2},\frac{3}{4}]$.
    Thus, all remaining arcs can directly host a demand of at most~$\frac{1}{2}$ and the remaining capacity of~$\frac{1}{2}$ has all be used to carry demand over paths of length exactly two.
    For the demand between~$u_i$ and~$v_i$, there are two possible ways and the vertices~$s$ and~$t$ will ensure that we can use one of the options exactly~$K$ times and the corresponding elements will have to be a cover in the original instance.

    We now prove the correctness formally and start with the simple direction, that is, we assume that there is an exact cover~$\mathcal{F}'$ for the original instance of \textsc{Exact Cover by 3-Sets}.
    For each pair~$(u,v)$ of vertices, we add~$g(d_{u,v})$ arcs from~$u$ to~$v$ and whenever~$g(d_{u,v}) \geq 1$, we add the pair~$(((u,v)),f(d_{u,v}))$ to the solution~$\mathcal{P}$.
    That is, we host the integer parts of all demands directly and for all remaining demands, if they are at least~$\frac{3}{4}$, then we also host them directly.
    We call the arcs added so far \emph{heavy arcs}.
    Note that the demands hosted directly by heavy arcs is~$H = \sum_{u \in V} \sum_{v \in V} f(d_{u,v})$.
    
    We have so far added~$\ns$ outgoing arcs for each vertex in \[\{t\} \cup \{y_i \mid i \in [N]\} \cup A \cup B \cup Z \cup \bigcup_{i \in [M]} W_i.\]
    Moreover, we have added~$\ns-1$ outgoing arcs for each vertex in~$\{u_i,v_i \mid i \in [N]\}$ and~$\ns-K$ outgoing arcs for~$s$.
    For each set~$S_i$ such that~$S_i \in \mathcal{F}'$, we add the arcs~$(s,u_i), (u_i,v_i)$, and~$(v_i,t)$ to the solution graph and the pairs~$(((s,u_i)),\frac{1}{2}), (((s,u_i),(u_i,v_i)),\frac{1}{2}), (((u_i,v_i)),\frac{1}{2})$, and~$(((v_i,t)),\frac{1}{2})$ to the solution flow~$\mathcal{P}$.
    For each set~$S_i$ such that~$S_i \notin \mathcal{F}'$, we add the arcs~$(u_i,v_i)$ and~$(v_i,s_i)$ to the solution graph and the pairs~$(((u_i,v_i)),\frac{1}{2}),(((u_i,v_i)),\frac{1}{2})$, and~$(((u_i,v_i),(v_i,u_i)),\frac{1}{2})$ to~$\mathcal{P}$.
    Note that for each vertex~$y_j$, we added a heavy arc to~$v_i$ for each~$i$ such that~$x_j \in S_i$ and sent a flow of~$\frac{5}{6}$ over this arc.
    Moreover, since~$\mathcal{F}'$ is an exact cover, for each~$x_j$, there exists exactly one set~$S_{i_j} \in \mathcal{F}'$ such that~$x_j \in S_{i_j}$.
    We finally add the pair~$(((y_j,v_{i_j}),(v_{i_j},t)),\frac{1}{6})$ to~$\mathcal{P}$ for each~$j \in [N]$.
    Note that the heavy arc~$(y_j,v_{i_j})$ has~$\frac{1}{6}$ unused capacity.
    This concludes the construction of our solution.

    It remains to show that the constructed graph is a directed~$\ns$-regular graph and that the constructed flow hosts a~$\kappa$-fraction of the constructed demand.
    As shown above, there are exactly~$\ns$ outgoing arcs for each vertex.
    It is also easy to verify that each vertex has exactly~$\ns$ incoming arcs in the constructed solution.
    The heavy arcs directly host~$H$ of the demand (scaled by~$\ns$).
    So it remains to show that the remaining~$\frac{3L}{4}+\frac{N}{12}$ demand (scaled by~$\ns$) are hosted by our construction, where~$L= n \ns - \sum_{u \in V}\sum_{v \in V}g(d_{u,v})$.
    To this end, note that we added exactly~$L$ non-heavy arcs.
    By construction, for each~$S_i \notin \mathcal{F}'$, we added two arcs~($(u_i,v_i)$ and $(v_i,u_i)$) and they together satisfy a demand of~$\frac{3}{2}$.
    For each set~$S_i \in \mathcal{F}'$, we added three arcs~($(s,u_i)$, $(u_i,v_i)$, and~$(v_i,t)$).
    The two arcs~$(s,u_i)$ and~$(u_i,v_i)$ satisfy a demand of~$\frac{3}{2}$ (scaled by~$\ns$) together.
    The final arc~$(v_i,t)$ hosts a demand of~$\frac{1}{2}$ directly and together with some of the unused capacity of heavy arcs going from~$y_j$ to~$v_{i_j}$, they satisfy another demand of~$3 \frac{1}{6} = \frac{1}{2}$.
    Thus, the total demand (scaled by~$\ns$) additionally hosted is~$\frac{3L}{4} + \frac{N}{12}$.
    The total demand hosted is therefore~$\frac{H+\frac{3L}{4}+\frac{N}{12}}{\ns} = \frac{\kappa}{n}$, that is, we hosted a~$\kappa$-fraction of all demands.
    Thus, the constructed instance of \thr{} is a yes-instance.

    In the other direction, we assume that the constructed instance of \thr{} is a yes-instance.
    We first show that that the contribution of all arcs is an equivalent measure for the weak throughput.
    To this end, let~$(G,\mathcal{P})$ be a solution, where~$\frac{\sum_{(P_i,d_i) \in \mathcal{P}} d_i}{n} \geq \kappa$.
    For the sake of notational convenience, we will set~$d'_i = \ns d_i$ and work with the equivalent assumption~$\sum_{(P_i,d_i) \in \mathcal{P}} d'_i \geq \kappa n \ns = H + \frac{3L}{4}+\frac{N}{12}$.
    Let~$E$ be the set of arcs in the constructed solution graph and let~$\con(e) = \sum_{(P_i,d_i) \in \mathcal{P}} \frac{d'_i}{\len(P)} \cdot \mathbf{1}_{e \in P_i}$ be the contribution of an arc~$e \in E$.
    Note that
    \[\sum_{e \in E} \con(e) = \sum_{e \in E}\sum_{(P_i,d_i) \in \mathcal{P}} \frac{d'_i}{\len(P)} \cdot \mathbf{1}_{e \in P_i} = \sum_{(P_i,d_i) \in \mathcal{P}} \sum_{e \in P_i} \frac{d'_i}{\len(P_i)} = \sum_{(P_i,d_i) \in \mathcal{P}} d'_i.\]
    Thus, the contribution of all arcs summed up is exactly the amount of demand that is hosted and we may assume that~$\sum_{e \in E} \con(e) \geq H+\frac{3L}{4} + \frac{N}{12}$.

    Next, note that the contribution of an arc~$e$ can only be larger than~$\frac{3}{4}$ if it hosts a demand larger than~$\frac{1}{2}$ directly.
    We will show that all such potential arcs have to be part of any solution.
    This implies that for any pair~$(u,v)$ of vertices, any solution contains~$g(d_{u,v})$ arcs from~$u$ to~$v$ and these host~$f(d_{u,v})$ demand directly, that is, all these arcs (henceforth called \emph{heavy arcs}) host a demand of~$H$ in total.
    The heavy arcs are exactly the same heavy arcs as in the proof of the forward direction and there are~$2M + K = n \ns - \sum_{u \in V} \sum_{v \in V} g(d_{u,v}) = L$ other arcs in the solution, where exactly one is outgoing from each vertex~$u_i$ and~$v_i$ and~$K$ are outgoing from~$s$.
    All these non-heavy arcs can host at most~$\frac{1}{2}$ demand directly since there is by construction no pair of vertices where the non-integer part of the demand is in the interval~$(\frac{1}{2},\frac{3}{4})$.
    This means that they each can contribute at most~$\frac{3}{4}$.
    So even if all heavy arcs host as much demand as possible directly (which maximizes their contribution per capacity used) and all other arcs contribute~$\frac{3}{4}$ each, the total contribution thus far is only~$H + \frac{3L}{4} = \kappa n \ns - \frac{N}{12}$.
    This means that some heavy arcs contribute more than what they can host directly.
    We will now slightly change the accounting and say that the contribution of a heavy arc is the amount what it can host directly and the contribution of all other arcs is~$\sum_{(P_i,d_i) \in \mathcal{P}} \frac{d'_i}{\len'(P)} \cdot \mathbf{1}_{e \in P_i}$ where~$\len'(P)$ is the number of non-heavy arcs in~$P$.
    Note that no additional demand can be served by just the unused capacity of heavy arcs and hence this re-accounting still precisely describes the weak throughput.
    We next show that the set of arcs contributing at least~$\frac{3}{4}$ did not change and the only arcs that can now contribute strictly more than~$\frac{3}{4}$ are arcs of the form~$(y_j,v_i)$.
    First, the only heavy arcs that have spare capacity are arcs of the form~$(s,a_i), (w_j^i,u_i), (y_j,v_i)$, and~$(b_i,t)$.
    Since none of the heavy arcs with spare capacity or of the potential non-heavy arcs are outgoing from vertices in~$\{t\} \cup A \cup Z$ or incoming into vertices in~$\bigcup_{i \in [M]}W_i$, there are only two possible types of demands that can now contribute more than before the re-accounting: arcs from a vertex in~$W_i$ for some~$i \in [M]$ to a vertex in~$Z$ and arcs from a vertex~$v_j$ to~$t$.
    Since no arc from a vertex in~$Z$ to a vertex outside~$Z$ can contribute~$\frac{3}{4}$, the former is also excluded.
    This implies that the total additional contribution is a flow of~$\frac{1}{6}$ from~$y_j$ to~$t$ for each~$j \in [N]$.
    These arcs cannot be direct arcs since their contribution could be at most~$\frac{1}{6} + \frac{5}{12} < \frac{3}{4}$ and thus, at most~$\frac{N}{6}$ flow is re-accounted.
    Since each such path contains at most~$1$ heavy arc and at least one non-heavy arc, the re-accouting for this flow adds an additional contribution of at most~$\frac{N}{12}$ to non-heavy arcs.
    This happens when each such flow uses a path~$P$ with~$\len(P) = 2$ and~$\len(P)'=1$.
    Thus, any potential solution must consist of all heavy arcs and all other arcs must contribute at least~$\frac{3}{4}$.

    We now analyze the structure of a potential solution and focus on the non-heavy arc leaving a vertex~$u_i$.
    Note that this arc contributes at least~$\frac{3}{4}$ and hence is not~$(u_i,u_i)$, but~$(u_i,v_i)$.
    This implies that~$u_i$ has exactly one non-heavy incoming arc.
    Since this arc also contributes at least~$\frac{3}{4}$, it has to also host a demand of at least~$\frac{1}{2}$ directly or only together with heavy arcs, that is, it is either~$(s,u_i)$ or~$(v_i,u_i)$.
    Note that the potential non-heavy outgoing arcs from~$s$ are~$(s,u_i)$ and~$(s,v_i)$ but each~$v_i$ already has~$\ns$ incoming arcs.
    Hence, $s$ has exactly~$K$ non-heavy outgoing arcs to~$K$ vertices~$u_{i_1},u_{i_2}, \ldots, u_{i_K}$.
    We will show that~$\mathcal{F}' = \{S_{i_1},S_{i_2}, \ldots, S_{i_K}\}$ is a solution for the original instance of \textsc{Exact Cover by 3-Sets}.
    For the corresponding~$K$ vertices~$u_i$, the two arcs~$(s,u_i)$ and~$(u_i,v_i)$ contribute exactly~$\frac{3}{4}$ each by satisfying the demands~$(s,u_i), (s,v_i)$ and~$(u_i,v_i)$.
    The corresponding vertices~$v_i$ all have a non-heavy arc~$(v_i,t)$ which hosts a demand of~$\frac{1}{2}$ directly.
    In order for them to also allow the additional flow of~$\frac{1}{6}$ from each vertex~$y_j$ to~$t$ to be hosted, each vertex~$y_j$ needs to have an outgoing arc to a node~$v_i$ that has an arc to~$t$.
    Since there are exactly~$K$ such arcs, each vertex~$v_i$ only has three incoming arcs from vertices~$y_j$, and~$N=3K$, it must hold that no two nodes~$v_{i_1}$ and~$v_{i_2}$ have incoming arcs from the same vertex~$v_j$, that is, the~$K$ sets in~$\mathcal{F}'$ do not intersect, that is, $\mathcal{F}'$ is an exact cover.
    This concludes the proof.
\end{proof}
\else
\begin{theorem}[$\star$]
    \label{thm:NPhard}
    \thr{} is NP-hard.
\end{theorem}
\fi

Finally, we conjecture \strong{} to be NP-hard as well, but we are currently not able to prove this.

\section{Conclusion}
\label{sec:conclusion}
We investigated four variants of demand-aware throughput.
We showed that two of them (based on direct arcs) can be solved optimally in polynomial time and gave essentially tight lower and upper bounds for them.
We also showed the first separation between demand-oblivious and demand-aware throughput in the more general case where demands can be routed via paths of arbitrary length---at least for a sufficiently large number of nodes.
This result is based on previous results in the context of randomized and approximation algorithms as well as the probabilistic method; the result is nevertheless completely deterministic.

We conclude with some open problems.
First, it would be nice to have tight bounds for the two remaining problem variants.
While our lower bound for \strong{} is a strong result as it achieves the first separation from the demand-oblivious construction, we do not think that it is optimal.
Hence, we ask for a different construction that gives better bounds and which also applies for any number~$n$ of nodes.
We also conjecture that it should be NP-hard to compute the demand-aware throughput, but since we were unable to prove this formally, we leave it as an open problem.
It would also be interesting to know what are the matrices that have the lowest demand-aware throughput.
In the case of demand-oblivious throughput, it is known that these are the permutation matrices.
A similar question also arises for the weak throughput.
Speaking of weak throughput, we conjecture that the lower bound for \thr{} should be improvable to something like~$\frac{8}{9}$.
A matching upper bound for~$n > 2$ is also still open.
Finally, to the best of our knowledge, demand-oblivious weak throughput has not been studied before: 
What is the best bound for demand-oblivious weak throughput? 

\ifarxiv
\section*{Acknowledgement}
\else
\begin{acks}
\fi
This work received financial support from the German Research Foundation (DFG) under research grant SPP 2378: ReNO-2 (project number 511099228) and from the Israel Science Foundation (ISF) grant no. 2497/23.
\ifarxiv
\else
\end{acks}
\fi

\ifarxiv
\bibliographystyle{plainnat}
\else
\bibliographystyle{ACM-Reference-Format}
\balance
\fi
\bibliography{bib}

@book{GJ79,
  author       = {Michael R. Garey and
                  David S. Johnson},
  title        = {Computers and Intractability: {A} Guide to the Theory of NP-Completeness},
  publisher    = {W. H. Freeman},
  year         = {1979},
  isbn         = {0-7167-1044-7}
}

@article{GKPS06,
  author       = {Rajiv Gandhi and
                  Samir Khuller and
                  Srinivasan Parthasarathy and
                  Aravind Srinivasan},
  title        = {Dependent rounding and its applications to approximation algorithms},
  journal      = {Journal of the {ACM}},
  volume       = {53},
  number       = {3},
  pages        = {324--360},
  year         = {2006},
  doi          = {10.1145/1147954.1147956}
}

@article{PS97,
  author       = {Alessandro Panconesi and
                  Aravind Srinivasan},
  title        = {Randomized Distributed Edge Coloring via an Extension of the {C}hernoff-{H}oeffding Bounds},
  journal      = {{SIAM} Journal on Computing},
  volume       = {26},
  number       = {2},
  pages        = {350--368},
  year         = {1997},
  doi          = {10.1137/S0097539793250767}
}

@book{CS,
  author       = {Rajeev Motwani and
                  Prabhakar Raghavan},
  title        = {Randomized Algorithms},
  publisher    = {Cambridge University Press},
  year         = {1995},
  doi2          = {10.1017/CBO9780511814075},
  isbn         = {0-521-47465-5},
}

@book{DBLP:books/networkflows,
  author       = {Ravindra K. Ahuja and
                  Thomas L. Magnanti and
                  James B. Orlin},
  title        = {Network flows -- {T}heory, algorithms and applications},
  publisher    = {Prentice Hall},
  year         = {1993},
  isbn         = {978-0-13-617549-0},
}

@article{addanki2025vermilion,
  author       = {Vamsi Addanki and
                  Chen Avin and
                  Goran Dario Knabe and
                  Giannis Patronas and
                  Dimitris Syrivelis and
                  Nikos Terzenidis and
                  Paraskevas Bakopoulos and
                  Ilias Marinos and
                  Stefan Schmid},
  title        = {Vermilion: {A} Traffic-Aware Reconfigurable Optical Interconnect with Formal Throughput Guarantees},
  journal      = {CoRR},
  volume       = {abs/2504.09892},
  year         = {2025},
  doi          = {10.48550/ARXIV.2504.09892},
  eprinttype    = {arXiv},
}

@article{greenberg2011vl2,
  author       = {Albert G. Greenberg and
                  James R. Hamilton and
                  Navendu Jain and
                  Srikanth Kandula and
                  Changhoon Kim and
                  Parantap Lahiri and
                  David A. Maltz and
                  Parveen Patel and
                  Sudipta Sengupta},
  title        = {{VL2: A} scalable and flexible data center network},
  journal      = {Communications of the {ACM}},
  volume       = {54},
  number       = {3},
  pages        = {95--104},
  year         = {2011},
  doi          = {10.1145/1897852.1897877},
}

@article{avin2020demand,
  author       = {Chen Avin and
                  Kaushik Mondal and
                  Stefan Schmid},
  title        = {Demand-aware network designs of bounded degree},
  journal      = {Distributed Computing},
  volume       = {33},
  number       = {3-4},
  pages        = {311--325},
  year         = {2020},
  doi          = {10.1007/S00446-019-00351-5},
}

@article{avin2022demand,
  author       = {Chen Avin and
                  Kaushik Mondal and
                  Stefan Schmid},
  title        = {Demand-Aware Network Design With Minimal Congestion and Route Lengths},
  journal      = {{IEEE/ACM} Transactions on Networking},
  volume       = {30},
  number       = {4},
  pages        = {1838--1848},
  year         = {2022},
  doi          = {10.1109/TNET.2022.3153586},
}

@article{zerwas2023duo,
  author       = {Johannes Zerwas and
                  Csaba Gy{\"{o}}rgyi and
                  Andreas Blenk and
                  Stefan Schmid and
                  Chen Avin},
  title        = {Duo: {A} High-Throughput Reconfigurable Datacenter Network Using Local Routing and Control},
  journal      = {Proceedings of the {ACM} on Measurement and Analysis of Computing Systems},
  volume       = {7},
  number       = {1},
  pages        = {20:1--20:25},
  year         = {2023},
  doi          = {10.1145/3579449},
}

@article{griner2021cerberus,
  author       = {Chen Griner and
                  Johannes Zerwas and
                  Andreas Blenk and
                  Manya Ghobadi and
                  Stefan Schmid and
                  Chen Avin},
  title        = {Cerberus: The Power of Choices in Datacenter Topology Design -- {A} Throughput Perspective},
  journal      = {Proceedings of the {ACM} on Measurement and Analysis of Computing Systems},
  volume       = {5},
  number       = {3},
  pages        = {38:1--38:33},
  year         = {2021},
  doi          = {10.1145/3491050},
}

@article{addanki2023mars,
  author       = {Vamsi Addanki and
                  Chen Avin and
                  Stefan Schmid},
  title        = {Mars: Near-Optimal Throughput with Shallow Buffers in Reconfigurable Datacenter Networks},
  journal      = {Proceedings of the {ACM} on Measurement and Analysis of Computing Systems},
  volume       = {7},
  number       = {1},
  pages        = {2:1--2:43},
  year         = {2023},
  doi          = {10.1145/3579312},
}

@article{mogul2012we,
  author       = {Jeffrey C. Mogul and
                  Lucian Popa},
  title        = {What we talk about when we talk about cloud network performance},
  journal      = {Computer Communication Review},
  volume       = {42},
  number       = {5},
  pages        = {44--48},
  year         = {2012},
  doi          = {10.1145/2378956.2378964},
}

@article{farrington2013multiport,
  title={A multiport microsecond optical circuit switch for data center networking},
  author={Farrington, Nathan and Forencich, Alex and Porter, George and Sun, P.-C. and Ford, Joseph E. and Fainman, Yeshaiahu and Papen, George C. and Vahdat, Amin},
  journal={IEEE Photonics Technology Letters},
  volume={25},
  number={16},
  pages={1589--1592},
  year={2013},
  doi={10.1109/LPT.2013.2270462}
}

@inproceedings{al2008scalable,
  author       = {Mohammad Al{-}Fares and
                  Alexander Loukissas and
                  Amin Vahdat},
  title        = {A scalable, commodity data center network architecture},
  booktitle    = {Proceedings of the {ACM} {SIGCOMM} 2008 Conference},
  pages        = {63--74},
  publisher    = {{ACM}},
  year         = {2008},
  doi          = {10.1145/1402958.1402967},
}

@inproceedings{li2019hpcc,
  author       = {Yuliang Li and
                  Rui Miao and
                  Hongqiang Harry Liu and
                  Yan Zhuang and
                  Fei Feng and
                  Lingbo Tang and
                  Zheng Cao and
                  Ming Zhang and
                  Frank Kelly and
                  Mohammad Alizadeh and
                  Minlan Yu},
  title        = {{HPCC: H}igh precision congestion control},
  booktitle    = {Proceedings of the {ACM} {SIGCOMM} 2019 Conference},
  pages        = {44--58},
  publisher    = {{ACM}},
  year         = {2019},
  doi          = {10.1145/3341302.3342085},
}

@inproceedings{dc-throughput,
  author       = {Pooria Namyar and
                  Sucha Supittayapornpong and
                  Mingyang Zhang and
                  Minlan Yu and
                  Ramesh Govindan},
  title        = {A throughput-centric view of the performance of datacenter topologies},
  booktitle    = {Proceedings of the {ACM} {SIGCOMM} 2021 Conference},
  pages        = {349--369},
  publisher    = {{ACM}},
  year         = {2021},
  doi          = {10.1145/3452296.3472913},
}

@inproceedings{faizian2017throughput,
  author       = {Peyman Faizian and
                  Md Atiqul Mollah and
                  Md. Shafayat Rahman and
                  Xin Yuan and
                  Scott Pakin and
                  Mike Lang},
  title        = {Throughput Models of Interconnection Networks: {T}he Good, the Bad,
                  and the Ugly},
  booktitle    = {Proceedings of the 25th {IEEE} Annual Symposium on High-Performance Interconnects ({HOTI})},
  pages        = {33--40},
  publisher    = {{IEEE} Computer Society},
  year         = {2017},
  doi          = {10.1109/HOTI.2017.21},
}

@inproceedings{jyothi2016measuring,
  author       = {Sangeetha Abdu Jyothi and
                  Ankit Singla and
                  Brighten Godfrey and
                  Alexandra Kolla},
  title        = {Measuring and understanding throughput of network topologies},
  booktitle    = {Proceedings of the 2016 International Conference for High Performance Computing, Networking, Storage and Analysis ({SC})},
  pages        = {761--772},
  publisher    = {{IEEE} Computer Society},
  year         = {2016},
  doi          = {10.1109/SC.2016.64},
}

@inproceedings{singla2014high,
  author       = {Ankit Singla and
                  Philip Brighten Godfrey and
                  Alexandra Kolla},
  title        = {High Throughput Data Center Topology Design},
  booktitle    = {Proceedings of the 11th {USENIX} Symposium on Networked Systems Design and Implementation ({NSDI})},
  pages        = {29--41},
  publisher    = {{USENIX} Association},
  year         = {2014},
}

@inproceedings{amir2022optimal,
  author       = {Daniel Amir and
                  Tegan Wilson and
                  Vishal Shrivastav and
                  Hakim Weatherspoon and
                  Robert Kleinberg and
                  Rachit Agarwal},
  title        = {Optimal oblivious reconfigurable networks},
  booktitle    = {Proceedings of the 54th Annual {ACM} {SIGACT} Symposium on Theory of Computing ({STOC})},
  pages        = {1339--1352},
  publisher    = {{ACM}},
  year         = {2022},
  doi          = {10.1145/3519935.3520020},
}

@inproceedings{baral2026universal,
  title={Universal Connection Schedules for Reconfigurable Networking},
  author={Baral, Shaleen and Kleinberg, Robert and Martin, Sylvan and Rogers, Henry and Wilson, Tegan and Zhang, Ruogu},
  booktitle={Proceedings of the 37th Annual {ACM-SIAM} Symposium on Discrete Algorithms ({SODA})},
  pages={4996--5026},
  publisher = {{SIAM}},
  year={2026},
  doi={10.1137/1.9781611978971.18}
}

@inproceedings{mellette2017rotornet,
  author       = {William M. Mellette and
                  Rob McGuinness and
                  Arjun Roy and
                  Alex Forencich and
                  George Papen and
                  Alex C. Snoeren and
                  George Porter},
  title        = {RotorNet: {A} Scalable, Low-complexity, Optical Datacenter Network},
  booktitle    = {Proceedings of the {ACM} {SIGCOMM} 2017 Conference},
  pages        = {267--280},
  publisher    = {{ACM}},
  year         = {2017},
  doi          = {10.1145/3098822.3098838},
}

@inproceedings{ahn2009hyperx,
  author       = {Jung Ho Ahn and
                  Nathan L. Binkert and
                  Al Davis and
                  Moray McLaren and
                  Robert S. Schreiber},
  title        = {HyperX: {T}opology, routing, and packaging of efficient large-scale networks},
  booktitle    = {Proceedings of the 2009 International Conference for High Performance Computing, Networking, Storage and Analysis ({SC})},
  pages        = {1--11},
  publisher    = {{ACM}},
  year         = {2009},
  doi          = {10.1145/1654059.1654101},
}

@inproceedings{jain2014maximizing,
  author       = {Nikhil Jain and
                  Abhinav Bhatele and
                  Xiang Ni and
                  Nicholas J. Wright and
                  Laxmikant V. Kal{\'{e}}},
  title        = {Maximizing Throughput on a Dragonfly Network},
  booktitle    = {Proceedings of the 2014 International Conference for High Performance Computing, Networking, Storage and Analysis ({SC})},
  pages        = {336--347},
  publisher    = {{IEEE} Computer Society},
  year         = {2014},
  doi          = {10.1109/SC.2014.33},
}

@inproceedings{poutievski2022jupiter,
  author       = {Leon Poutievski and
                  Omid Mashayekhi and
                  Joon Ong and
                  Arjun Singh and
                  Muhammad Mukarram Bin Tariq and
                  Rui Wang and
                  Jianan Zhang and
                  Virginia Beauregard and
                  Patrick Conner and
                  Steve D. Gribble and
                  Rishi Kapoor and
                  Stephen Kratzer and
                  Nanfang Li and
                  Hong Liu and
                  Karthik Nagaraj and
                  Jason Ornstein and
                  Samir Sawhney and
                  Ryohei Urata and
                  Lorenzo Vicisano and
                  Kevin Yasumura and
                  Shidong Zhang and
                  Junlan Zhou and
                  Amin Vahdat},
  title        = {Jupiter evolving: {T}ransforming {G}oogle's datacenter network via optical circuit switches and software-defined networking},
  booktitle    = {Proceedings of the {ACM} {SIGCOMM} 2022 Conference},
  pages        = {66--85},
  publisher    = {{ACM}},
  year         = {2022},
  doi          = {10.1145/3544216.3544265},
}

@inproceedings{curtis2012rewire,
  author       = {Andrew R. Curtis and
                  Tommy Carpenter and
                  Mustafa Elsheikh and
                  Alejandro L{\'{o}}pez{-}Ortiz and
                  Srinivasan Keshav},
  title        = {{REWIRE: A}n optimization-based framework for unstructured data center network design},
  booktitle    = {Proceedings of the 31st Annual Joint Conference of the IEEE Computer and Communications Societies ({INFOCOM})},
  pages        = {1116--1124},
  publisher    = {{IEEE}},
  year         = {2012},
  doi          = {10.1109/INFCOM.2012.6195470},
}

@inproceedings{yuan2014lfti,
  author       = {Xin Yuan and
                  Santosh Mahapatra and
                  Michael Lang and
                  Scott Pakin},
  title        = {{LFTI:} {A} New Performance Metric for Assessing Interconnect Designs for Extreme-Scale {HPC} Systems},
  booktitle    = {Proceedings of the 28th {IEEE} International Parallel and Distributed Processing Symposium ({IPDPS})},
  pages        = {273--282},
  publisher    = {{IEEE} Computer Society},
  year         = {2014},
  doi          = {10.1109/IPDPS.2014.38},
}

@inproceedings{KCKL05,
  author       = {Isaac Keslassy and
                  Cheng{-}Shang Chang and
                  Nick McKeown and
                  Duan{-}Shin Lee},
  title        = {Optimal load-balancing},
  booktitle    = {Proceedings of the 24th Annual Joint Conference of the {IEEE} Computer and Communications Societies ({INFOCOM})},
  pages        = {1712--1722},
  publisher    = {{IEEE}},
  year         = {2005},
  doi          = {10.1109/INFCOM.2005.1498452},
}

@inproceedings{helios,
  author       = {Nathan Farrington and
                  George Porter and
                  Sivasankar Radhakrishnan and
                  Hamid Hajabdolali Bazzaz and
                  Vikram Subramanya and
                  Yeshaiahu Fainman and
                  George Papen and
                  Amin Vahdat},
  title        = {Helios: {A} hybrid electrical/optical switch architecture for modular data centers},
  booktitle    = {Proceedings of the {ACM} {SIGCOMM} 2010 Conference},
  pages        = {339--350},
  publisher    = {{ACM}},
  year         = {2010},
  doi          = {10.1145/1851182.1851223},
}

@inproceedings{firefly,
  author       = {Navid Hamed Azimi and
                  Zafar Ayyub Qazi and
                  Himanshu Gupta and
                  Vyas Sekar and
                  Samir R. Das and
                  Jon P. Longtin and
                  Himanshu Shah and
                  Ashish Tanwer},
  title        = {FireFly: {A} reconfigurable wireless data center fabric using free-space optics},
  booktitle    = {Proceedings of the {ACM} {SIGCOMM} 2014 Conference},
  pages        = {319--330},
  publisher    = {{ACM}},
  year         = {2014},
  doi          = {10.1145/2619239.2626328},
}

@inproceedings{projector,
  author       = {Monia Ghobadi and
                  Ratul Mahajan and
                  Amar Phanishayee and
                  Nikhil R. Devanur and
                  Janardhan Kulkarni and
                  Gireeja Ranade and
                  Pierre{-}Alexandre Blanche and
                  Houman Rastegarfar and
                  Madeleine Glick and
                  Daniel C. Kilper},
  title        = {{ProjecToR: A}gile Reconfigurable Data Center Interconnect},
  booktitle    = {Proceedings of the {ACM} {SIGCOMM} 2016 Conference},
  pages        = {216--229},
  publisher    = {{ACM}},
  year         = {2016},
  doi          = {10.1145/2934872.2934911},
}

@article{hall2021survey,
  author       = {Matthew Nance Hall and
                  Klaus{-}Tycho Foerster and
                  Stefan Schmid and
                  Ramakrishnan Durairajan},
  title        = {A Survey of Reconfigurable Optical Networks},
  journal      = {Optical Switching and Networking},
  volume       = {41},
  pages        = {100621},
  year         = {2021},
  doi          = {10.1016/J.OSN.2021.100621},
}

@article{avin2025revolutionizing,
  author       = {Chen Avin and
                  Stefan Schmid},
  title        = {Revolutionizing Datacenter Networks via Reconfigurable Topologies},
  journal      = {Communications of the {ACM}},
  volume       = {68},
  number       = {6},
  pages        = {44--53},
  year         = {2025},
  doi          = {10.1145/3708980},
}

@article{BAS26,
  author       = {Matthias Bentert and
                  Chen Avin and
                  Stefan Schmid},
  title        = {A Separation Between Optimal Demand-Oblivious and Demand-Aware Network Throughput},
  journal      = {CoRR},
  volume       = {XXX},
  year         = {2026},
  doi          = {XXX},
  eprinttype   = {arXiv},
  eprint       = {XXX}
}
\end{document}